\documentclass{article}
\usepackage{amsmath}
\usepackage{amsthm}
\usepackage{graphicx}
\usepackage[labelformat=simple]{subfig}
\usepackage{float}
\usepackage{verbatim}
\usepackage{xcolor}
\usepackage{complexity}
\usepackage{dsfont}
\newclass{\threesat}{3SAT}
\usepackage[inline]{enumitem}
\usepackage{todonotes}

\newcommand{\TAR}{\operatorname{TAR}}
\newcommand{\GPG}{\operatorname{GP}}
\newcommand{\TARC}[1][\alpha]{\operatorname{TAR-#1}}
\newcommand{\MS}{\operatorname{MS}}
\newtheorem{theorem}{Theorem}
\newtheorem{definition}[theorem]{Definition}
\newtheorem{lemma}[theorem]{Lemma}

\newtheorem{proposition}[theorem]{Proposition}

\usepackage[mathlines]{lineno}

\usepackage{url,hyperref}
\hypersetup{
	hidelinks,
	colorlinks=true,
	citecolor=[rgb]{0.121 0.47 0.705},
	linkcolor=[rgb]{0.121 0.47 0.705},
	urlcolor=[rgb]{0.121 0.47 0.705}
}

\def\inst#1{$^{#1}$}

\begin{document}

\title{Graphs with large total angular resolution}

\author{Oswin Aichholzer\inst{1}
\and Matias Korman\inst{2}
\and Yoshio Okamoto\inst{3}
\and Irene Parada\inst{4} 
\and  Daniel Perz\inst{1}
\and Andr{\'{e}} van Renssen\inst{5}
\and Birgit Vogtenhuber\inst{1} 
}

\date{}

\maketitle

\begin{center}
\inst{1}Graz University of Technology, Graz, Austria \\
\{oaich,daperz,bvogt\}@ist.tugraz.at
\\\ \\
\inst{2}Siemens EDA (previously Mentor Graphics), Wilsonville, USA \\
matias.korman@tufts.edu
\\\ \\
\inst{3}The University of Electro-Communications, Tokyo, Japan \\
okamotoy@uec.ac.jp
\\\ \\
\inst{4}Technical University of Denmark, Lyngby, Denmark \\
irmde@dtu.dk
\\\ \\
\inst{5}The University of Sydney, Sydney, Australia \\
andre.vanrenssen@sydney.edu.au
\end{center}

\begin{abstract}
The total angular resolution of a straight-line drawing is the minimum angle between two edges of the drawing.
It combines two properties contributing to the readability of a drawing: 
the angular resolution, which is the minimum angle between incident edges, 
and the crossing resolution, which is the minimum angle between crossing edges.
We consider the total angular resolution of a graph, 
which is the maximum total angular resolution of a straight-line drawing of this graph. 

We prove tight bounds for the number of edges for graphs for some values of the total angular resolution up to a finite number of well specified exceptions of constant size.
In addition, we show that deciding whether a graph has total angular resolution at least $60^{\circ}$ is \NP-hard.
Further we present some special graphs and their total angular resolution.
\end{abstract}

\section{Introduction}

\noindent We study angles between incident edges of straight-line drawings of graphs. In the following we mostly omit the word straight-line.
The \emph{total angular} resolution of a drawing $D$, or short $\TAR(D)$, is the smallest angle occurring in $D$, either between two edges incident to the same vertex or between two crossing edges.
In other words, $\TAR(D)$ is the minimum of the angular resolution $\operatorname{AR}(D)$ and the crossing resolution $\operatorname{CR}(D)$ of the same drawing (where $\operatorname{CR}(D)=360^{\circ}$ if $D$ is plane).
Furthermore, the total angular resolution 
of a graph $G$ (or short $\TAR(G)$) is defined as the maximum of $\TAR(D)$ 
over all drawings $D$ of $G$.
Similarly, the angular resolution and the crossing resolution of $G$ are the maximum of $\operatorname{AR}(D)$ and $\operatorname{CR}(D)$, respectively, over all drawings $D$ of $G$.
Note that the total angular resolution of a graph can be smaller than the minimum of its crossing resolution and its angular resolution, see Figure~\ref{fig:TAR:ar_cr_tar}. 

\begin{figure}[hbt]
	\centering 
	\subfloat[Drawing with \\ $\operatorname{AR}(D)=67.5^{\circ}$.\label{fig:TAR:AR_TAR}]{
		\includegraphics{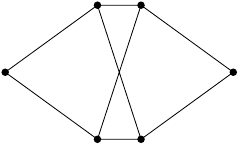}
	}
	\hspace{6mm}
	\subfloat[Drawing with $\operatorname{CR}(D)=90^{\circ}$.\label{fig:TAR:CR_TAR}]{
		\includegraphics{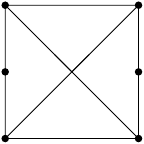}
	}
	\hspace{6mm}
	\subfloat[Drawing with $\TAR(D)=60^{\circ}$.\label{fig:TAR:TAR_TAR}]{
		\includegraphics{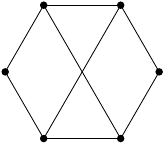}
	}
\caption{Three slightly different drawings $D$ of a graph $G$ with $\TAR(G)=60^{\circ}$.}
\label{fig:TAR:ar_cr_tar}
\end{figure}

In 1993 Formann et al.~\cite{Formann1993DrawingGI} were the first to introduce the angular resolution of graphs. They showed that finding a drawing of a graph with angular resolution at least $90^{\circ}$ is \NP-hard for graphs with maximum degree $4$.
Four yeary later, Malitz and Papakostas~\cite{MP94_ARPG} studied the angular resolution of planar graphs.
For more results about the angular resolution see for example \cite{DK03_PCDPGGAR,duncan2013drawing,KKLMRV21_SRG}.
Another eleven years later, experiments by Huang et al.~\cite{Huang2007UsingET,Huang2008EffectsOC} showed that the crossing resolution plays a major role in the readability of drawings.
Consequently, research in that direction was intensified.
In particular, right angle crossing drawings (or short RAC drawings) were introduced by Didimo, Eades and Liotta~\cite{didimo2011drawing}.
Van Kreveld \cite{van2010quality} showed among other results for RAC drawings, that the angular resolution in RAC drawings can be by an arbitrary factor larger than the angular resolution in plane drawings.
The \NP-hardness of deciding whether a given drawing admits an RAC drawing was proven by Argyriou, Bekos and Symvonis~\cite{Argyriou2011TheSR}.
For a recent survey by Didimo on RAC drawings, see~\cite{Didimo2020racdg}.

For $\alpha$AC drawings (drawings with crossing resolution $\alpha$), Dujmovi\'{c} et al.~\cite{Dujmovic2010NotesOL} showed an upper bound on the number of edges of $\frac{180^{\circ}}{\alpha}(3n-6)$.
For the two special classes of RAC drawings and $\alpha$AC drawings with $60^{\circ}<\alpha<83^{\circ}$ better upper bounds are known~\cite{Ackerman2007OnTM}.
More precisely, Didimo et al. showed that RAC drawings have at most $4n-10$ edges~\cite{didimo2011drawing} and that this bound is tight.
 For $\alpha$AC drawings with ${60^{\circ}<\alpha<83^{\circ}}$, Ackerman and Tardos proved an upper bound of at most ${6.5n-20}$~edges~\cite{Ackerman2007OnTM}. This bound is due to the fact that quasi-planar drawings (drawings without three pairwise crossing edges) have at most ${6.5n-20}$~edges and drawings that are not quasi-planar have a crossing angle of at most~${60^{\circ}}$.

Argyriou, Bekos and Symvonis~\cite{Argyriou2013MaximizingTT} were the first to study the total angular resolution, calling it just \emph{total resolution}. They presented drawings of complete graphs and complete bipartite graphs with asymptotically optimal total angular resolution.
Recently, Bekos et al.~\cite{Bekos2018AHA} presented a new heuristic for finding a drawing of a given graph with high total angular resolution which performed superior to earlier heuristics like~\cite{Argyriou2013MaximizingTT,huang2013improving} on the considered test cases.
For a recent survey on angular resolution, crossing resolution and total angular resolution see Okamoto~\cite{Okamoto2020}.
\newline

\subsection*{Outline}

\noindent In this work we show that almost all graphs with $\TAR(G)>60^{\circ}$ have at most $2n-6$ edges, list the finitely many such graphs that have more than $2n-6$ edges and show that this bound is tight. 
Moreover, we show the following tight
upper bounds on the number of edges for graphs with larger total angular resolution:
$2n-2\sqrt{n}$ for 
${\TAR(G)\geq 90^{\circ}}$, 
$\frac{3}{2}n-\frac{5}{2}$ for 
${\TAR(G)>90^{\circ}}$ (and $n\geq 3$), 
and $n$ for 
$\TAR(G)>120^{\circ}$ (and $n\geq 7$). 
We also prove that it is \NP-hard to determine whether $\TAR(G)\geq 60^{\circ}$. 

Further, we present an infinite family of graphs with $\TAR(G)=60^{\circ}$, for which every proper subgraph $G'$ of $G$ has $\TAR(G')>60^{\circ}$.
We conclude this work with a bound for the number of edges that can be removed from the complete graph $K_n$ without changing its total angular resolution.

\section{Upper bounds on the number of edges}
\label{sec:upperbounds}

\noindent In this section we study the relation between the total angular resolution and the maximal number of edges.
First we need some definitions.
Every straight-line drawing $D$ (of a graph $G$) partitions the plane into connected regions which are called \emph{cells} of $D$.
The \emph{planarization} $P(D)$ of a drawing $D$ is the drawing in which we replace every crossing by a vertex so that this new vertex splits both crossing edges into two edges. 
Furthermore, every edge in $P(D)$ has two \emph{sides} and every side is incident to exactly one cell of $P(D)$.
Note that both sides of an edge can be incident to the same cell.
A connected drawing $D$ is a drawing such that $P(D)$ corresponds to a drawing of a plane connected graph.
We define the \emph{size} of a cell of a connected drawing $D$ as the number of sides in $P(D)$ incident to this cell.

\subsection{Graphs with $\TAR(G)>60^{\circ}$}

\noindent In this section we show that for almost all graphs with $\TAR(G)>60^{\circ}$ the number of edges is bounded by $2n-6$. We start by showing a bound for the number of edges in a connected drawing $D$ depending on the size of the unbounded cell of~$D$.

\begin{lemma}\label{lemma:main}
Let $D$ be a connected drawing with $n\geq 1$ vertices, $m$ edges and $\TAR(D)>60^{\circ}$. If the unbounded cell of $D$ has size $k$,
then $m \leq 2 n -2 -\left\lceil {k}/{2} \right\rceil$.
\end{lemma}

\begin{proof}
If at least three edges cross each other in a single point, then there exists an angle with at most $60^{\circ}$ at this crossing point.
Therefore, every crossing of the drawing $D$ is incident to exactly two edges.
We planarize $D$ and get 
$n_P= n + \operatorname{cr}(D)$ and $m_P= m + 2 \operatorname{cr}(D)$
where $\operatorname{cr}(D)$ is the number of crossings in $D$, $n_P$ is the number of vertices of $P(D)$, and $m_p$ is the number of edges of~$P(D)$.
Since $P(D)$ is a plane drawing, we can use Euler's formula to compute the number $f_p$ of faces in $P(D)$ as
\begin{equation}
f_P = -n + m + \operatorname{cr}(D) + 2. \label{euler_face}
\end{equation}
Moreover, every bounded cell of $D$ has size at least $4$, as otherwise $P(D)$ would contain a triangle, implying an angle of at most $60^{\circ}$.
By definition, the unbounded cell of $D$ has size $k$ and we obtain the inequality
\begin{equation}
2 m_P \geq 4(f_P - 1) + k.  \label{lemma:counting}
\end{equation}
	Combining Equation~(\ref{euler_face}) and Inequality~(\ref{lemma:counting}) with $m_P= m + 2 \operatorname{cr}(D)$ gives  \newline 
	${m \leq 2 n -2 -\left\lceil {k}/{2} \right\rceil}$.
\end{proof}

From Lemma~\ref{lemma:main} it follows directly that a connected drawing $D$ on $n \geq 3$ vertices and with $\TAR(D)>60^{\circ}$  fulfills $m\leq 2n-4$.
Note that this bound is only two edges away from the optimal upper bound. 

We proceed to prove the bound of $2n-6$ for disconnected drawings.
\begin{lemma}\label{lemma:disconnected}
Let $D$ be a disconnected drawing on $n \geq 3$ vertices with \linebreak $\TAR(D) > 60^{\circ}$.
Then, $m\leq 2n-6$ or $D$ consists of three vertices and one edge (Exception E0 in Figure~\ref{fig:exceptionsAll}).
\end{lemma}

\begin{proof}
Assume that $D$ consists of $\ell \geq 2$ components $C_i$, $1\leq i \leq \ell$, where $C_i$ has $n_i \geq 1$ vertices and $m_i \geq 0$ edges.
Furthermore, $\TAR(C_i)\geq \TAR(D)> 60^{\circ}$ holds.
By Lemma~\ref{lemma:main} we get $m_i\leq 2n_i-2$ for every component.
If $\ell\geq 3$, this directly implies 

\begin{equation*}
m=\sum_{i=1}^{\ell} m_i \leq \sum_{i=1}^{\ell} (2 n_i -2) = 2n - 2 \ell \leq 2n-6.
\end{equation*}
Now consider $\ell=2$.
If $C_1$ contains at least 2 edges, then the size of the unbounded cell of $C_1$ is at least $3$.
So we get $m_1\leq 2n_1 - 4$ by Lemma~\ref{lemma:main}.
This gives

\begin{equation*}
m=m_1+m_2 \leq 2n_1 - 4 + 2n_2-2=2n - 6.
\end{equation*}
This implies that we only have to check drawings with exactly two connected components and each has at most one edge.
If both $C_1$ and $C_2$ consist of two vertices and one edge, then we have
 ${m=2=2\cdot 4 -6}$~edges. 
If $D$ is a drawing on $n=3$ vertices with $m=1$ edges, then we have Exception E0.
\end{proof}

Now we prove the bound for connected drawings.
To further improve the bound from Lemma~\ref{lemma:main}, the following lemma is useful.

\begin{lemma}\label{obs:inner_degree}
Let $D$ be a plane connected drawing where the boundary of the unbounded cell is a simple polygon $P$ 
with $p>3$ vertices. 
Let the inner degree of a vertex~$v_i$ of $P$ be the number~$d'(v_i)$ of edges incident to $v_i$ that lie in the interior of~$P$.
If $\TAR(D)>60^{\circ}$, then $\sum_{v_i \in V(P)} d'(v_i) \leq 2 p - 7$ holds.
\end{lemma}

\begin{proof}
Assume to the contrary that $\TAR(D)>60^{\circ}$ and $\sum_{v_i \in V(P)} d'(v_i) \geq 2 p - 6$.
The sum of all inner angles in any simple polygon $P$ with $p$ vertices is $180^{\circ}(p-2)$.
The number of all inner angles in $D$ incident to vertices in $P$ is $p+ \sum_{v_i \in V(P)} d'(v_i) \geq 3 p-6$.
Since every angle in $D$ is larger than $60^{\circ}$, all inner angles incident to vertices of $P$ sum up to strictly more than $180^{\circ}(p-2)$.
This contradicts that the sum of the inner angles is $180^{\circ}(p-2)$.
Therefore, we have $\sum_{v_i \in V(P)} d'(v_i) \leq 2 p - 7$.
\end{proof}

\begin{lemma}\label{lemma:2n-5}
Let $D$ be a connected plane drawing on $n\geq 3$ vertices, where $D$ is not a path on $3$ vertices and not a $4$-cycle.
If $\TAR(D)>60^{\circ}$, then $m\leq 2n - 5$.
\end{lemma}

\begin{proof}
The unbounded cell of $D$ cannot have size $3$, as in this case the convex hull of the drawing is a triangle and we have $\TAR(D)\leq 60^{\circ}$.
If the drawing $D$ has an unbounded cell of size at least $5$ and $\TAR(D)>60^{\circ}$, then $m \leq 2n -5$ follows directly from Lemma~\ref{lemma:main}.
Otherwise, the unbounded cell of $D$ has size~$4$, which, as $D$ is not a path on $3$ vertices, implies that the boundary of $D$ is a 4-cycle $F$. By Lemma~\ref{obs:inner_degree} and the fact that $D$ is not a 4-cycle, $D$ contains exactly one edge $e$ with one vertex in the interior of $F$ and the other vertex on the boundary of $F$.
Let $D'$ be the drawing we obtain by deleting all vertices and edges of $F$ and also the edge $e$. 
The drawing $D'$ is connected and has $n' \geq 1$ vertices and $m'$ edges, where $n'=n-4$ and $m'=m-5$. 
By Lemma~\ref{lemma:main} we know that $m' \leq 2 n' - 2$ and we derive
$m = m'+5 \leq 2 n' -2 + 5 = 2n - 5$.
\end{proof}

Two drawings are \emph{combinatorially equivalent} if all cells are bounded by the same edges, all crossing edge pairs are the same, and for each edge $e$ the order of crossings along $e$ is the same. We extend Lemma~\ref{lemma:2n-5} in the following way.

\begin{lemma}\label{lemma:2n-6}
	Let $D$ be a connected plane drawing on $n \geq 3$ vertices with $m$~edges and $\TAR(D) > 60^{\circ}$. 
	Then, $m\leq 2n-6$ unless $D$ is combinatorially equivalent to one
	of the exceptions E1--E9 listed below 
	and depicted in Figure~\ref{fig:exceptionsAll}.\\[-4ex]
	\begin{enumerate} [label={\emph{\bf E\arabic*}}]
\item 
	A tree on at most 4 vertices.
\item  
	A plane $4$-cycle.
\item 
	A plane $4$-cycle with one additional vertex connected to one vertex of the $4$-cycle. The vertex can be inside or outside the cycle.
\item  
	A plane $5$-cycle.
\item 
	A plane $5$-cycle with one vertex inside connected to two non-neighboring vertices of the $5$-cycle.
\item  
	A plane $5$-cycle with an edge inside, connected with 3 edges to the 5-cycle
  such that the interior of the 5-cycle is partitioned into two 4-faces and one 5-faces.
\item 
	A plane $6$-cycle with an additional diagonal between opposite vertices.
\item 
	A plane $6$-cycle with an additional vertex or edge inside, connected with 3 or 4, respectively, edges to the 6-cycle
  such that the interior of the 6-cycle is partitioned into 3 or 4, respectively, 4-faces.
\item  
	A plane $6$-cycle with either a path on 3 vertices or a 4-cycle inside, connected with 5 edges to the 6-cycle such that the interior of the $6$-cycle is partitioned into 4 or 5, respectively, 4-faces.
\end{enumerate}
\end{lemma}

\begin{figure}[ht]
\centering 
\includegraphics[page=2]{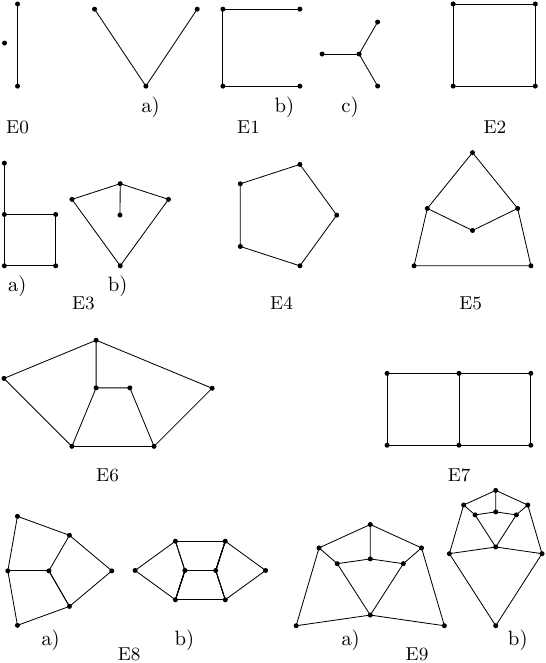}
\caption{All exceptions for Lemma~\ref{lemma:2n-6} and Theorem~\ref{theorem:2n-6}}\label{fig:exceptionsAll}
\end{figure}

\begin{proof}
Let $D'$ be a subdrawing of $D$ consisting of all vertices that are not on the unbounded cell and of all edges that are not incident to a vertex on the unbounded cell. 
Assume $D'$ has $n'$ vertices and $m'$ edges.
We distinguish four cases, depending on the size of the unbounded cell.
\begin{description}
\item[Case 1] The unbounded cell has size $4$.
If the drawing has only one cell, then it is Exception~E1a.
Otherwise, the boundary of the unbounded cell is a 4-cycle $C$ and we have $n=n'+4$. 
As by Lemma~\ref{obs:inner_degree}, there is at most one edge in $D$ from a vertex of $C$ to $D'$, we have $m\leq m'+5$.

If there is at most one vertex inside $C$, then we have Exception E2 or E3b. So assume that there are at least two vertices inside $C$.
Since there is at most one edge from a vertex of $C$ to the interior and $D'$ is connected, $D'$ thus has at least one edge.
So the unbounded cell of $D'$ has size at least~$2$.
 Hence, by Lemma~\ref{lemma:main} for $D'$, it holds that ${m'\leq 2n'-3}$ and we obtain

\begin{equation*}
m\leq m'+5\leq 2n'-3 +5 = 2(n-4)+2 = 2 n - 6.
\end{equation*}

\item[Case 2] The unbounded cell has size $5$.
In this case, the outer boundary must be a $5$-cycle: The only other possibility would be a triangle with an attached edge, but in that case we would have $\TAR(D)\leq 60^{\circ}$. 
Hence, we have $n'=n-5$.
If there are at most two adjacent vertices inside the 5-cycle that are connected with edges to the 5-cycle, then we have one of the Exceptions E4, E5, or E6. So assume that there are at least 3 vertices in the interior. Due to Lemma~\ref{obs:inner_degree}, there are at most three edges connecting the interior to the $5$-cycle and the $5$-cycle itself has $5$ edges, that is, $m\leq m'+5+3=m'+8$.
If $D'$ is connected, then the size of the unbounded cell of $D'$ is at least $3$ and we have $m'\leq 2n'-4$ by Lemma~\ref{lemma:main}.
Otherwise $D'$ consists of two or three connected components.
By Lemma~\ref{lemma:disconnected} we have $m'\leq 2n'-6$ unless $D'$ consists of three vertices and an edge, which gives $m'\leq 2n'-5$, or it contains fewer than three vertices.
The only disconnected drawing with fewer than three vertices is a drawing consisting of two vertices, which gives $m'\leq 2n'-4$.
So, in all cases we get $m'\leq 2n'-4$ and hence have

\begin{equation*}
m \leq m' + 8 \leq 2 n' - 4 +8 = 2 n - 6.
\end{equation*}

\item[Case 3] The unbounded cell of the drawing $D$ has size $6$.
If $D$ has only one cell (i.e. only the unbounded cell), we have Exception E1b or E1c. Otherwise the boundary $B$ of the unbounded cell of $D$ either consists of two triangles sharing a vertex ($\TAR(D)\leq 60^{\circ}$) or is a $4$-cycle with an attached edge or a $6$-cycle. So there are two cases we have to consider.
\begin{itemize}
\item If $B$ is a $4$-cycle with an attached edge, we use similar arguments as in Case 1. If there is no vertex inside the $4$-cycle, then we have Exception E3a. If we have at least one point inside the 4-cycle, then by Lemma~\ref{lemma:main} we have $m'\leq 2 n' - 2$. 
So we get

\begin{equation*}
m = m' + 6 \leq 2 n' - 2 + 6 = 2(n'+5)-6 = 2n - 6.
\end{equation*}

\item If $B$ is a $6$-cycle, then by Lemma~\ref{obs:inner_degree} we can have at most $5$ edges connecting the interior to the 6-cycle.
First we consider the case that $D'$ is connected.
If $\TAR(D)>60^{\circ}$ and $n'\geq 3$, then $\TAR(D')>60^{\circ}$ and $D'$ fulfills ${m' \leq 2n'-5}$ by Lemma~\ref{lemma:2n-5} unless $D'$ is a path on 3 vertices or a 4-cycle. 
Furthermore, we know that $n=n'+6$ and $m\leq m'+11$.
If $m' \leq 2 n' - 5$, then

\begin{equation*}
m \leq m' + 11 \leq 2 n' - 5 +11 = 2n-6.
\end{equation*}

Consider now the case that $n'\leq 2$ or $D'$ is a path on 3 vertices or a 4-cycle. 
These cases can be checked by hand.
Therefore, we have Exceptions E7 and E8 if $n'\leq 2$, and Exceptions E9 if $D'$ is a path on 3 vertices or a 4-cycle. 

If $D'$ is not connected and $\TAR(D)>60^{\circ}$, then 
 either $D'$ fulfills $m'\leq 2n'-6$, or $D'$ consists of three vertices and an edge (by Lemma~\ref{lemma:disconnected}), or $D'$ consists of two vertices.
If $D'$ fulfills $m'\leq 2n'-6$ or consists of three vertices and an edge, then we have $m \leq 2n-6$.
So consider the case that $D'$ consists of two vertices.
This means that one of the two inner vertices has degree at least $3$ in the drawing $D$. If one vertex has degree 4, then there is a triangle in our drawing $D$ which means that ${\TAR(D)\leq 60^{\circ}}$. Otherwise, if one vertex has degree 3 and the other one has degree 2, then we have a drawing like in Figure~\ref{figure:ex6gon_2points}. 
The gray shaded \text{$4$-cycle} has 2 edges in the interior. 
So due to Lemma~\ref{obs:inner_degree} we have ${\TAR(D)\leq 60^{\circ}}$.

\begin{figure}[tb]
\centering 
	\includegraphics[scale=0.9]{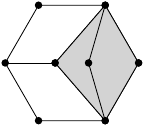}
\caption{Two separated vertices inside a $6$-cycle.}\label{figure:ex6gon_2points}
\end{figure}
\end{itemize}

\item[Case 4] The unbounded cell has size at least $7$.
Then we have, by Lemma~\ref{lemma:main},
\begin{equation*}
m \leq 2 n - 2 - \left\lceil \frac{k}{2} \right\rceil = 2 n - 2 - \left\lceil \frac{7}{2} \right\rceil \leq 2 n - 6.
\end{equation*}
\end{description}
\end{proof}

Note that Lemma~\ref{lemma:2n-6} considers plane drawings. 
Next we consider drawings with at least one crossing, whose planarizations are in the exceptions of Lemma~\ref{lemma:2n-6}.
If $D$ has a crossing, then $P(D)$ has a vertex of degree at least~$4$.
The only exceptions with such a vertex are the ones of E9; see again Figure~\ref{fig:exceptionsAll}.

\begin{lemma}
	If we replace the vertex of degree $4$ in a drawing of E9 in Figure~\ref{fig:exceptionsAll}
	with a crossing, then the resulting drawings $D$ have $\TAR(D)\leq 60^{\circ}$.
\end{lemma}
\begin{proof}
\begin{figure}[H]
\centering 
\includegraphics{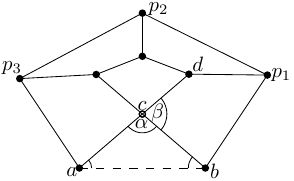}
	\caption{Replacing the vertex of degree $4$ of the drawing of E9a in Figure~\ref{fig:exceptionsAll} with a crossing.}\label{figure:ex_cross_proof}
\end{figure}

If we replace the vertex of degree $4$ of Exception E9a in Figure~\ref{fig:exceptionsAll} with a crossing,
 then we get the drawing $D_{cr}$ in Figure~\ref{figure:ex_cross_proof}, where the dashed edge is not part of the actual drawing. 
We want to show that $\TAR(D_{cr})\leq 60^{\circ}$.
The crossing edge pair forms two angles. As indicated in Figure~\ref{figure:ex_cross_proof}, we denote $\angle acb$ as $\alpha$ and $\angle bcd$ as $\beta$, where $c$ denotes the crossing and $a$, $b$ and $d$ are three of the endpoints of the crossing edges.
Let $p_1,p_2$ and $p_3$ be the other three vertices on the unbounded cell.
Since $c$ is a crossing, $c$ is inside the pentagon $abp_1p_2p_3$.
The inner angles of a pentagon sum up to $540^{\circ}$.
All eight inner angles of the drawing, that are incident to the pentagon $abp_1p_2p_3$, are larger than $60^{\circ}$.
This implies that $\angle bac + \angle abc < 60^{\circ}$.
Furthermore we have $\alpha + \beta = 180^{\circ} = \alpha + \angle bac + \angle abc$.
This means we have $\beta = \angle bac + \angle abc < 60^{\circ}$.
However, $\beta$ appears in $D_{cr}$, and so we have $\TAR(D_{cr})\leq 60^{\circ}$.

Now let $D'_{cr}$ be the drawing we get if we replace the vertex of degree $4$ in the drawing for E9b in Figure~\ref{fig:exceptionsAll} with a crossing.
Then $D_{cr}$ is a subdrawing of~$D'_{cr}$ and hence we get $\TAR(D'_{cr})\leq \TAR(D_{cr})\leq 60^{\circ}$.
\end{proof}

So we have characterized all drawings $D$ which have $\TAR(D)>60^{\circ}$ and ${m>2n-6}$ edges, such that $P(D)$ is in the exceptions of Lemma~\ref{lemma:2n-6}.
This leads us to the following theorem.

\begin{theorem}\label{theorem:2n-6}
Let $G$ be a graph with $n\geq 3$ vertices, $m$ edges and ${\TAR(G)>60^{\circ}}$. 
	Then $m \leq 2n - 6$ unless either there exists a drawing of $G$ that is an exception for Lemma~\ref{lemma:2n-6}
	or
	$G$ consists of exactly three vertices and one edge 
	(Exception E0 in Figure~\ref{fig:exceptionsAll}). 
	Further, if $G$ is a graph that forms an exception for Lemma~\ref{lemma:2n-6}, then every drawing $D$ of $G$ is drawn plane and combinatorially equivalent to an exception of Lemma~\ref{lemma:2n-6}.
\end{theorem}

\begin{proof}
Consider a graph $G$ with $n\geq 3$ vertices, $m>2n-6$ edges and $\TAR(G)>60^{\circ}$ 
Then there exists a drawing $D$ of $G$ with $\TAR(D)>60^{\circ}$ and its planarization~$P(D)$.

If $G$ is disconnected, then by Lemma~\ref{lemma:disconnected} it has either $m\leq 2n-6$ edges or consists of three vertices and one edge.
So for the rest of the proof we only consider connected graphs.

If three edges cross in a single point, then in $P(D)$ this point has degree $6$, and therefore an angle with at most $60^{\circ}$.
Hence every crossing involves exactly two edges and $P(D)$
	has $m_P= m + 2\operatorname{cr}(D)$ edges and $n_P=n+\operatorname{cr}(D)$ vertices.	
By Lemma~\ref{lemma:2n-6} we get that $m_P\leq 2n_P -6$ or $P(D)$ is in the exceptions.
If $m_P\leq 2n_P -6$, then $m=m_P - 2\operatorname{cr}(D) \leq  2(n_P-\operatorname{cr}(D))-6=2n-6$.
If $P(D)$ is in the exceptions, then, as observed before, $D$ is in the exceptions.
\end{proof}

The bound of Theorem~\ref{theorem:2n-6} is the best possible in the sense that there are infinitely many graphs with $m=2 n - 6$ edges and $\TAR(G)>60^{\circ}$.
\begin{proposition}\label{prop:tight}
For every integer $n\geq 17$ there exists a graph $G$ with $n$ vertices and $m=2n-6$ edges such that $\TAR(G)>60^{\circ}$.
\end{proposition}
\begin{proof}

\begin{figure}
\centering
\includegraphics{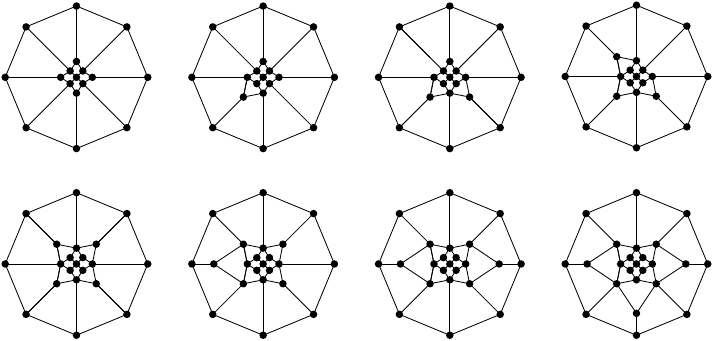}
\caption{Drawings of graphs $G$ with $\TAR(G) > 60^{\circ}$, $n$ vertices and $m=2n-6$ edges for $17\leq n\leq 24$}\label{fig:tightness}
\end{figure}

Figure~\ref{fig:tightness} illustrates drawings $D$ with $17\leq n \leq 24$ vertices, $m=2n-6$ edges and $\TAR(D)>60^{\circ}$.
We extend this family of drawings, such that for any number of vertices $n\geq 17$ we have a drawing with $m=2n-6$ edges and $\TAR(D)>60^{\circ}$, by adding layers of 8-cycles as illustrated in Figure~\ref{fig:tight_extension}.

\begin{figure}
\centering
\includegraphics{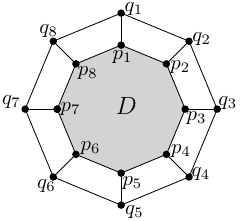}
\caption{Extending the drawings in Figure~\ref{fig:tightness}}\label{fig:tight_extension}
\end{figure}

Let $D$ be a drawing with $n$ vertices, $m=2n-6$ edges, $\TAR(D)>60^{\circ}$ and  whose boundary is a regular $8$-cycle $C=p_1p_2p_3p_4p_5p_6p_7p_8$. 
We construct a bigger drawing $D'$ in the following way.
Let $C'=q_1q_2q_3q_4q_5q_6q_7q_8$ be a regular $8$-cycle, which is concentric to $C$. Further, for any $1\leq i \leq 8$ let $p_i$ and $q_i$ be on a ray from the common circumcenter of $C$ and $C'$.
We merge $D$ with the $8$-cycle~$C'$ and all edges $p_iq_i$ with $1\leq i \leq 8$ and call the resulting drawing $D'$.

First observe that since $\TAR(D)>60^{\circ}$ also $\TAR(D')>60^{\circ}$ holds, because the angles $\angle p_{i-1}q_{i-1}q_{i}=\angle q_{i-1}q_ip_i=67.5^{\circ}$ and $\angle q_{i}p_{i}p_{i-1}=\angle p_{i+1}p_iq_i=112.5^{\circ}$ with $q_0=q_8$ and $p_0=p_8$ for every $1\leq i \leq 8$.
The drawing $D'$ contains ${n'=n+8}$~vertices and $m'=m+16$ edges ($m$ edges of $D$, 8 edges on $C'$ and 8 edges connecting $C$ and $C'$). Since $D$ has $n$ vertices and $m=2n-6$ edges, also $m'=m+16=2n-6+16=2(n+8)-6=2n'-6$.
So by extending a drawing $D$ to $D'$ in this way we get eight more vertices.
Since every drawing depicted in Figure~\ref{fig:tightness} has a regular $8$-cycle as boundary, we are able to extend each of these drawings as described before. Doing this repeatedly, we are able to add $8k$ vertices to each of the drawings for every integer $k>0$.
The numbers of vertices of the drawings depicted in Figure~\ref{fig:tightness} cover all parities modulo eight. So there exists for any number $n\geq 17$ a graph $G$ with $\TAR(G) > 60^{\circ}$ and $m=2n-6$ edges.
\end{proof}

Note that in Proposition~\ref{prop:tight} we constructed plane graphs with $n$~vertices, $2n-6$~edges such that $\TAR(G)>60^{\circ}$.
In the following proposition we show, that there are drawings with $n$~vertices, $2n-6$~edges and $k$~crossings such that $\TAR(G)>60^{\circ}$.

\begin{proposition}
For every integer $k$, there exists a drawing $D$ with $n$~vertices, ${2n-6}$~edges and $k$~crossings such that $\TAR(D)=60^{\circ}$.
\end{proposition}

\begin{proof}
Consider the drawing $D$ in Figure~\ref{fig:tight_crossings}.
It consists of a point $c$, a regular $7$-gon $P={p_1p_2p_3p_4p_5p_6p_7}$ with circumcenter $c$, the center $p_0$ of the line segment $p_1p_7$,
 another regular $7$-gon $Q={q_1q_2q_3q_4q_5q_6q_7}$ with circumcenter $c$, such that $c$, $p_i$ and $q_i$ are on a line for $1\leq i \leq 7$ and the center $q_0$ of the line segment $q_1q_7$.
Further, we have the edge $cp_0$, $cp_2$, $cp_4$ and $cp_6$.
We have another $7$-gon in Figure~\ref{fig:tight_crossings} whose points are on the line segment $p_iq_i$ for $1\leq i \leq 7$ and that crosses the line segment $p_0q_0$.
Observe, that $\TAR(D)>64^{\circ}$
\begin{figure}
\centering
\includegraphics{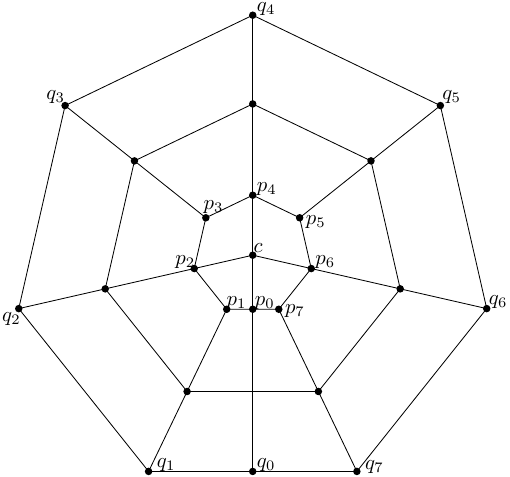}
\caption{Drawing with $24$ vertices $2n-6=42$ edges and one crossing.}\label{fig:tight_crossings}
\end{figure}

Instead of adding just one $7$-gon whose points are on the line segment $p_iq_i$ for $1\leq i \leq 7$, we add $k$ of these $7$-gons to get an extended drawing $D'$.
Each of these crosses the line segment $p_0q_0$ once and they do not induce other crossings.
Further the respective line segments of the $7$-gons are parallel.
Hence, $\TAR(D')=\TAR(D)>64^{\circ}$.

Therefore, $D'$ has $24+7k$ vertices, $42+14k$ edges and $k$ crossing, with $\TAR(D')=\TAR(D)>64^{\circ}$.
\end{proof}

\subsection{Graphs with $\TAR(G)\geq 90^{\circ}$}

\noindent Bodlaender and Tel~\cite{Bodlaender2004ANO} showed that if a graph can be embedded with angular resolution of at least $90^{\circ}$, then the graph can also be embedded such that all angles at vertices have one of the values $90^{\circ}, 180^{\circ}, 270^{\circ}$, and $360^{\circ}$.
Note that in any such drawing, 
the angle between two crossing edges is exactly $90^{\circ}$. Hence, by~\cite{Bodlaender2004ANO}, an angular resolution of at least $90^{\circ}$ for a graph $G$ implies ${\TAR(G)\geq 90^{\circ}}$.
In this section, we show that graphs with $\TAR(G)\geq 90^{\circ}$ have at most ${\lfloor 2n-2\sqrt{n}\rfloor}$~edges, which is tight.

\begin{lemma}\label{lemma:90:lower}
	For every $n \geq 1$, there exists a graph $G$ with $n$ vertices, $\lfloor 2n-2\sqrt{n}\rfloor$ edges, and $\TAR(G) = 90^{\circ}$.
\end{lemma}

\begin{proof}
We will construct the graph $G$ along with a drawing $D$ for $G$ that shows $\TAR(G) = 90^{\circ}$.
If we take a square grid with $k$ vertices on each side, then we have in total $n=k^2$ vertices and $m=2k^2-2k$ edges.
So for a $k\times k$ grid we have $m=\lfloor 2n-2\sqrt{n}\rfloor$, which proves the statement for $n=k^2$.

\begin{figure}[hbt]
	\centering 
	\subfloat[Grid drawing with $m=k^2+r$ edges, $1\leq r \leq k$.\label{fig:TAR90:k2+r1}]{
		\includegraphics{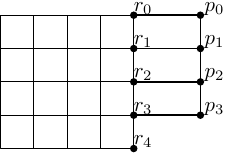}
	}
	\hspace{20mm}
	\subfloat[Grid drawing with $m=k^2+r$ edges, $k+1\leq r< 2k+1$.\label{fig:TAR90:k2+r2}]{
		\includegraphics{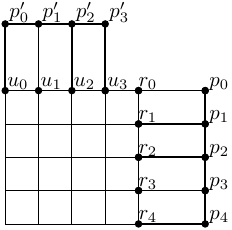}
	}
\caption{Different extensions of a grid drawing}
\label{fig:TAR90:square}
\end{figure}

For $k^2<n<(k+1)^2$ we extend the graph in the following way; see Figure~\ref{fig:TAR90:square}.
We call the rightmost points of the grid $r_0, r_1, \dots, r_{k-1}$ from top to bottom and the topmost points $u_0,\dots, u_{k-1}$ from left to right. (Note that ${r_0=u_{k-1}}$.)

If $n=k^2+r$ with $0< r < k$, we place additional points $p_0,\dots, p_{r-1}$, all on a common vertical line and each $p_i$ is to the right of $r_i$ for $0\leq i \leq r-1$, as depicted in Figure~\ref{fig:TAR90:k2+r1} where the added edges are drawn heavier than the old ones.
Further, we add the edges $r_ip_i$, $0\leq i < r$, and $p_{j-1}p_{j}$, $1\leq j < r$.
This gives us $r+(r-1)=2r -1$ new edges.
So we have ${m=2k^2 -2k +2r-1}$~edges in total.
On the other hand we get

\begin{align*}
\lfloor 2n-2\sqrt{n}\rfloor &= \lfloor 2(k^2 +r)-2\sqrt{k^2+r}\rfloor \\
                            &= 2k^2 +2r + \lfloor - 2\sqrt{k^2+r}\rfloor \\
                            &= 2k^2 - 2k +2r -1= m
 \end{align*}
because $k<\sqrt{k^2 + r} < k + \frac{1}{2}$
if  $1\leq r \leq k$.
So $m=\lfloor 2n-2\sqrt{n}\rfloor$ if $n=k^2+r$ and $1\leq r \leq k$.

For $n=k^2+k+r$, with $1\leq r < k+1$, we add $k$ points as before and also add the same edges.
The remaining $r$ points $p'_0,\dots,p'_r$ are placed on a horizontal line, such that $p'_i$ is above $u_i$ for $0\leq i \leq r-1$.
We further add the edges $p'_iu_i$, $0\leq i \leq r-1$, and $p'_{j-1}p'_j$, $1\leq j\leq r-1$, as depicted in Figure~\ref{fig:TAR90:k2+r2} where the added edges are drawn heavier than the old ones.
So we have $m=2k^2+2r-2$ edges, which is again $m=\lfloor 2n-2\sqrt{n}\rfloor$ for $n=k^2+k+r$ and $1\leq r \leq k$.
This means that for every $n$ there exists a graph $G$ with $\TAR(G)= 90^{\circ}$, $n$~vertices and ${\lfloor 2n - 2\sqrt{n}\rfloor}$~edges.
\end{proof}

\begin{lemma}\label{lemma:90:upper}
Every graph with $\TAR(G)= 90^{\circ}$ has at most ${\lfloor 2n - 2\sqrt{n}\rfloor}$ edges.
\end{lemma}

\begin{proof}
We prove the statement by contradiction.
Let $G$ be a graph with $n$~vertices and $m$~edges.
By Bodlaender and Tel~\cite{Bodlaender2004ANO} we can embed our graph, such that every angle is $90^{\circ},180^{\circ},270^{\circ}$ or $360^{\circ}$. So we can embed our graph on some rectangular grid $R$ with $a\times b$ points, such that in every column and every row there is at least one point and such that the edges are along the grid. We call this drawing $D$.

Now we add edges and vertices to $D$, so that this new drawing $D'$ is the complete $a\times b$ grid.

\begin{figure}[H]
	\centering 
	\subfloat[Example of a drawing $D$ with an underlying grid $R$.\label{fig:TAR90ub:base}]{
		\includegraphics[page=1, scale=1]{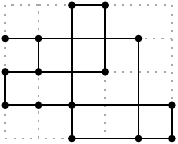}
	}
	\hfill
	\subfloat[Replacing a crossing of $D$ with a vertex.\label{fig:TAR90ub:crossing}]{
		\includegraphics[page=2, scale=1]{figures/TAR_90/extending_to_grid}
	}
	\hfill
	\centering 
	\subfloat[$D'$ after replacing every crossing.\label{fig:TAR90ub:cross_all}]{
		\includegraphics[page=3, scale=1]{figures/TAR_90/extending_to_grid}
	}
	
	\centering
	\subfloat[Adding corners of $R$ to $D'$.\label{fig:TAR90ub:corner}]{
		\includegraphics[page=4]{figures/TAR_90/extending_to_grid}
	}
		\hfill 
	\subfloat[Adding missing edges on the boundary to $D'$.\label{fig:TAR90ub:border}]{
		\includegraphics[page=5]{figures/TAR_90/extending_to_grid}
	}
	
	\centering
	\subfloat[Hitting a vertex $u$ along $\ell$: Add $uv$ to $D'$.\label{fig:TAR90ub:hit_vertex}]{
		\includegraphics[page=6]{figures/TAR_90/extending_to_grid}
	}
	\hfill
	\subfloat[Hitting an edge $xy$ along $\ell$: Add $z$ and $vz$ to $D'$ and split $xy$.\label{fig:TAR90ub:hit_edge}]{
		\includegraphics[page=7]{figures/TAR_90/extending_to_grid}
	}
	\hfill
	\subfloat[Completing $D'$ to a full grid.\label{fig:TAR90ub:full}]{
		\includegraphics[page=8]{figures/TAR_90/extending_to_grid}
	}
\caption{Adding vertices and edges to get from an initial drawing $D$ to a drawing of a grid. The underlying grid $R$ is drawn in gray. The added vertices and edges are marked with red and all split edges are drawn in blue for each step.}
\label{fig:TAR90:adding}
\end{figure}

At the beginning we set $D'=D$, see Figure~\ref{fig:TAR90ub:base}.
First we add a vertex at every crossing of~$D'$ as depicted in Figure~\ref{fig:TAR90ub:crossing} and Figure~\ref{fig:TAR90ub:cross_all}.
By doing this we get four edges instead of two and one new vertex.
So we get two new edges and one new vertex for each crossing.

In the second step, we add corner vertices of $R$ to $D'$. 
If all four corner vertices are already in the drawing, then we skip this step.
Without loss of generality assume that the top left corner vertex of the grid is not a vertex in the drawing.
We call the topmost vertex in the leftmost column $l_t$ and the leftmost vertex in the topmost row $t_l$ as depicted in Figure~\ref{fig:TAR90ub:corner}.
Then we add the top left corner vertex $v_{tl}$ of the grid together with the edges $v_{tl}l_t$ and $v_{tl}t_l$.
Analogously, we add the remaining missing corner vertices.
So for every added corner vertex we also added two edges to $D'$.

Next, we add edges between two points on the boundary of $R$ as illustrated in Figure~\ref{fig:TAR90ub:border}, so that the outer face of $D'$ is a rectangle with possibly some vertices on its sides. Since we already added the corner vertices in the second step, we only add edges in this step.

In the last step, we check if every vertex has full degree (degree 4 for inner vertices, degree 3 for vertices on the boundary, which are not corners, and degree~2 for corner vertices).
Assume we have a vertex $v$ without full degree.
Then there is a line segment $\ell$ of $R$ such that no edge of $v$ is along $\ell$.
We draw a line segment $e$ from $v$ along $\ell$ until we hit a vertex or an edge of $D'$.
If we hit a vertex $u$ as depicted in Figure~\ref{fig:TAR90ub:hit_vertex}, then we add the edge $uv$ to $D'$. In this case, we do not add any vertex.
If we hit an edge $xy$ as depicted in Figure~\ref{fig:TAR90ub:hit_edge}, then we add a vertex $z$, where we hit the edge, to $D'$.
Further we add the edge~$vz$ to $D'$ and split $xy$ into $xz$ and $yz$ in $D'$.
In this case we have one additional vertex and increased the number of edges by two in $D'$.

If every vertex has full degree, then $D'$ is equal to an $a \times b$ grid.
Further, every time we added a vertex, we also increased the number of edges by two.
So this new drawing $D'$ has $n'=n+k$ vertices and $m'$~edges where ${m'\geq m+2k}$.
On the other hand, an $a\times b$ rectangle grid has exactly $n'=ab$ vertices and ${m'=2ab-(a+b)}$~edges.
So we have 

\begin{align*}
m'&=2ab-(a+b)\leq 2ab-2\sqrt{ab}=2(n+k)-2\sqrt{n+k} \text{, and hence} \\
m &\leq m'-2k \leq 2n - 2\sqrt{n+k} \leq 2n -2\sqrt{n}.
\end{align*}
This means that every graph $G$ with $\TAR(G) = 90^{\circ}$ and $n$ vertices has at most ${{\lfloor 2n-2\sqrt{n} \rfloor}}$ edges.
\end{proof}

The example for a graph with $n$ vertices, ${{\lfloor 2n-2\sqrt{n} \rfloor}}$ edges and $\TAR(G) = 90^{\circ}$ is planar.
With some modification of the corresponding drawing we are also able to constuct drawings with $k$ crossing, $n$ vertices, ${{\lfloor 2n-2\sqrt{n} \rfloor}}$ edges and $\TAR(D) = 90^{\circ}$ for any integer $k$.

\begin{proposition}
For every integer $k$ there exists a drawing $D$ with $n$~vertices, ${{\lfloor 2n-2\sqrt{n} \rfloor}}$~edges and $k$~crossings such that $\TAR(D)=90^{\circ}$.
\end{proposition}

\begin{proof}
Let $D$ be a $k+2 \times k+2$ grid where $k$ inner points are replaced by crossings.
This is possible since $P(D)$ contains $k^2\geq k$ inner points.
Obviously, $\TAR(D)=90^{\circ}$.
Then $P(D)$ has $k^2+4k+4$ vertices and $2k^2+6k+4$ edges.
Hence, $D$ has $k^2+3k+4$ vertices and $2k^2+4k+4$ edges.
Due to $k^2+3k+4>\left(k+\frac{3}{2}\right)^2$, we have
\begin{align*}
\lfloor 2n-2\sqrt{n} \rfloor &= \lfloor 2(k^2+3k+4)-2\sqrt{k^2+3k+4} \rfloor \\
&= 2(k^2+3k+4) - \lceil 2\sqrt{k^2+3k+4} \rceil \\
&= 2(k^2+3k+4) - (2k+4)  = 2k^2+4k+4.
\end{align*}
Therefore, $D$ has $n=k^2+3k+4$ vertices, ${{\lfloor 2n-2\sqrt{n} \rfloor}=2k^2+4k+4}$~edges and $k$ crossings such that $\TAR(D)=90^{\circ}$.
\end{proof}

\subsection{Graphs with $\TAR(G)> 90^{\circ}$}

\noindent If a graph has $\TAR(G)> 90^{\circ}$, then this graph is planar, since a crossing would imply that at least one angle is at most  $90^{\circ}$.
Also note that the construction for a graph with $\TAR(G)= 90^{\circ}$ and ${\lfloor 2n-2\sqrt{n} \rfloor}$ edges heavily relied on $4$-cycles.
So we can improve the bound for graphs with $\TAR(G)> 90^{\circ}$.
\begin{theorem}
Every graph $G$ with $n\geq 3$ vertices and $\TAR(G)>90^{\circ}$ has at most $\frac{3}{2}n-\frac{5}{2}$ edges.
This bound is tight for infinitely many values of $n$.
\end{theorem}

\begin{proof}
We observe that every vertex of a graph $G$ with $\TAR(G)>90^{\circ}$ has degree at most $3$.
This already gives an upper bound of $\frac{3}{2}n$ edges for graphs with $\TAR(G)>90^{\circ}$.
Let $D$ be a drawing of $G$ with $\TAR(D)>90^{\circ}$.
Then every vertex on the boundary of the convex hull of $D$ has degree at most $2$.
Further, consider the angles spanned by the convex hull edges of $D$.
Assume that this angle is at most $90^{\circ}$ for some convex hull vertex $v$.
If $v$ was incident to two edges of $D$, then these edges would span an angle of at most $90^{\circ}$.
So $v$ has degree at most $1$.

If there are at least $5$ vertices on the convex hull of $D$, then $D$ has at most $(n-5)$ vertices of degree $3$ and at least $5$ vertices of degree at most $2$. Therefore, $D$ has at most $\frac{3}{2}(n-5)+\frac{2}{2}\cdot 5=\frac{3}{2}n-\frac{5}{2}$ edges.

If there are exactly $4$ vertices on the convex hull of $D$, then at least one of the inside angles of the boundary of the convex hull is at most $90^{\circ}$.
Therefore, at least one vertex on the convex hull has degree $1$.
So $D$ has at most ${\frac{3}{2}n-\frac{5}{2}}$ edges.
Similarly, if there are exactly $3$ vertices on the convex hull of $D$, then at least two vertices of those have degree $1$.
Again $D$ has at most $\frac{3}{2}n-\frac{5}{2}$ edges. 
This means that every graph $G$ with at least $3$ vertices and $\TAR(G)>90^{\circ}$ has at most $\frac{3}{2}n-\frac{5}{2}$ edges.

\begin{figure}[hbt]
\centering
\subfloat[Drawing with $\TAR(D)=96^{\circ}$, $n=35$ vertices and $\frac{3}{2}n-\frac{5}{2}$ edges.\label{fig:TAR96_flower2}]{
		\includegraphics[page=1]{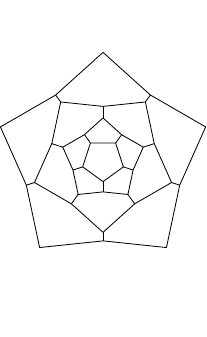}
	}
	\hfill
\subfloat[Extending Figure~\ref{fig:TAR96_flower2}. The gray lines are for construction purposes and $D$ is inside the cyan area.\label{fig:TAR96:expand}]{
		\includegraphics[page=2]{figures/TAR96_flower}
	}
\caption{Drawings with $\TAR(D)>90^{\circ}$}
\label{fig:TAR96_flower}
\end{figure}

Let $D$ be a drawing with $\TAR(D)\geq 96^{\circ}$, $n$ vertices and $\frac{3}{2}n-\frac{5}{2}$ edges such that the boundary $B$ of the convex hull of $D$ is a regular $5$-gon $p_1p_2p_3p_4p_5$, also illustrated in Figure~\ref{fig:TAR96:expand}.
Note, that a regular $5$-gon has these properties.
Let $c$ be the circumcenter of $B$ and let $K$ be a circle with center $c$ such that $D$ is inside $K$.
We call the crossing of the ray $cp_i$ with $K$ $q'_i$, for every $1\leq i \leq 5$.
The rays $cp_i$, $1\leq i \leq 5$, and the circle $K$ are gray in Figure~\ref{fig:TAR96:expand}.
The vertex $p'_i$ is the unique vertex with $\angle p_iq'_ip'_i=96^{\circ}$ and $\angle p'_{i}q'_{i+1}p_{i+1}=96^{\circ}$ for every $1\leq i \leq 5$ where $q'_6=q'_1$ and $p_6=p_1$.
Let $D'$ be $D$ together with the vertices $q'_i$ and $p'_i$ and the edges $p_iq'_i, q'_ip'_i$ and $p'_iq'_{i+1}$ for every $1\leq i \leq 5$.
By definition we have $\angle p_iq'_ip'_i=\angle p'_iq'_{i+1}p_{i+1}=96^{\circ}$.
Looking at the quadrilateral $cq'_ip'_iq'_{i+1}$ we have $\angle q'_ip'_iq'_{i+1}=96^{\circ}$.
Further the angles $\angle q'_{i+1}p_{i+1}p_{i}$ and $\angle p_{i+1}p_iq'_i$ are both inside an angle in $D'$ and both have $126^{\circ}$.
So we have $\TAR(D')=96^{\circ}$ with $n'=n+10$ vertices and $\frac{3}{2}n'-\frac{5}{2}$ edges.

Starting with a regular $5$-gon and doing this extension iteratively, there exists a drawing with $\TAR(D)=96^{\circ}$ with $n=10k+5$ vertices and $\frac{3}{2}n-\frac{5}{2}$ edges for every $k\geq 0$. 
For example, by doing this extension three times we get the drawing depicted in Figure~\ref{fig:TAR96_flower2}.
\end{proof}

\subsection{Graphs with $\TAR(G)>120^{\circ}$}

\noindent For angles $\alpha>120^{\circ}$ we prove a family of tight bounds for the number of edges of graphs $G$ with $\TAR(G)>\alpha$.
\begin{theorem}
Let $k>6$.
Every graph $G$ with $n$ vertices and $\TAR(G)\geq \frac{k-2}{k}180^{\circ}$ has at most $n$ edges for $n\geq k$, and at most $n-1$ edges otherwise.
These bounds are tight.
\end{theorem}

\begin{proof}
If a graph $G$ has $\TAR(G)>120^{\circ}$ then every vertex has degree at most two.
Any such graph is a collection of cycles, paths and isolated vertices.
So this graph has at most $n$ edges.
If $n\geq k$, then a regular $n$-gon $D$ has $\TAR(D)=\frac{n-2}{n}180^{\circ}>120^{\circ}$ and exactly $n$ edges.
This means that a graph $G$ with $\TAR(G)\geq \frac{k-2}{k}180^{\circ}$ has at most $n$ edges for $n\geq k$ and this bound is tight.

If $n<k$, then any cycle would prevent $\TAR(G)\geq \frac{k-2}{k}180^{\circ}$. This means a graph with $\TAR(G)\geq \frac{k-2}{k}180^{\circ}$ is cycle-free and has at most $n-1$ edges. 
On the other hand, if $G$ is a single path, then $\TAR(G)=180^{\circ}$. So a graph with $\TAR(G) \geq \frac{k-2}{k}180^{\circ}$ and $n<k$ vertices can have at most $n-1$ edges and this bound is tight.
\end{proof}

\section{NP-hardness}

\noindent Formann et al.~\cite{Formann1993DrawingGI} showed that the problem of determining whether a given graph $G$ admits a drawing with angular resolution of $90^{\circ}$ is \NP-hard. 
Their proof, which is by reduction from \threesat\ with exactly three different literals per clause, also implies the \NP-hardness of deciding whether $\TAR(G)=90^{\circ}$.
We adapt in the following their reduction to show the \NP-hardness of the decision problem for $\TAR(G)\geq 60^{\circ}$.

Note that every triangle of a drawing $D$ must be equilateral if $\TAR(D)\geq 60^{\circ}$ 
The idea of the construction is to build a rigid frame with triangles and add the clause gadgets such that they are also rigid; see Figure~\ref{fig:60hard:construction} for depictions of the frame and the gadgets.
Then, we add variable gadgets to the frame, such that they can only be oriented in two ways, which will correspond to the variable assignment.

\begin{theorem}\label{thm:hard}
It is \NP-hard to decide whether a graph $G$ has $\TAR(G)\geq 60^{\circ}$.
\end{theorem}

\begin{proof}
As input we are given a \threesat\ formula with variables $x_1,x_2,\dots,x_n$ and clauses $c_1, c_2,\dots,c_m$, where every clause contains exactly three different literals. Cook~\cite{Cook1971Complexity} showed that the decision question for satisfiability of such a \threesat\ formula is \NP-complete.

We first construct a graph $G$ for the formula.
The basic building blocks of our construction consist of triangles, which, in order to obtain a total angular resolution of $60^{\circ}$, must all be equilateral.
We use the following gadgets; see Figure~\ref{fig:60hard:gadget}.

As clause gadget we use a sequence of four triangles that share a common vertex and in which consecutive triangles share an edge.
The middle vertex with three incident edges, marked with $C_j$ in the figure, will 
be used to connect the clause gadget to its literals. We refer to $C_j$ as the \emph{clause vertex}.

As variable gadget we use a triangle followed by a sequence of $m$ hexagons and followed by another triangle.
Each hexagon consists of six triangles sharing the center point. 
Each non-extreme hexagon of the sequence is incident to its neighboring hexagons via two ``opposite'' edges.
The initial triangle is incident to the first hexagon via the edge opposite to the incidence with the second hexagon.
The final triangle is incident to the last hexagon via the edge opposite to the incidence with the second to last hexagon.
The vertices of the initial and the final triangle that are incident to none of the hexagons are denoted as $A_{i,1}$ and $A_{i,2}$, respectively.

For each variable $x_i$, we assign one side of the hexagonal path to the positive literal $x_i$ and the other to the negative literal $\overline{x}_i$. 
The intermediate vertices of the $j$th hexagon of the path are  
denoted with $X_{i,j}$ and $\overline{X}_{i,j}$, respectively, and are called \emph{literal vertices}.
 They will  
be used for connecting a literal to its clause.  

Additionally, we use a connector gadget. It consists of two triangles with a common edge.
The two vertices that are incident to only one of the triangles are denoted by $A_{i,3}$ and $A_{i,4}$, respectively.

Note that for all three gadgets, an embedding with total angular resolution $60^{\circ}$ is unique up to rotation, scaling and reflection of the whole gadget.
Especially, for each gadget, all triangles are congruent.  

\begin{figure}[hbt]
	\centering 
	\subfloat[All used gadgets\label{fig:60hard:gadget}]{
		\includegraphics[page=2]{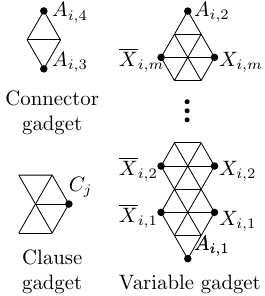}
	}
	\hfill
	\subfloat[Frame with clause gadgets\label{fig:60hard:basic_structure}]{
		\includegraphics[page=2]{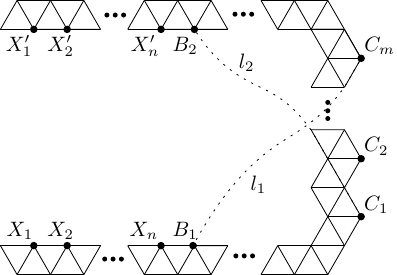}
	}
	\caption{Gadgets and frame of the \NP-hardness proof.}\label{fig:60hard:construction}
\end{figure}

For connecting the gadgets, we first build a rigid 3-sided frame as depicted in Figure~\ref{fig:60hard:basic_structure}.
On the bottom, it consists of a straight path of $2n+2m-1$ triangles that alternatingly face up and down (the \emph{bottom path}).
On top of the rightmost triangle of this path, we add a sequence of $m$ clause gadgets stacked on top of each other  
(one for each clause, with the clause vertices $C_1, \ldots, C_m$ facing to the right). 
The top of the figure consists of a straight path of ${2n+2m-1}$~triangles that alternatingly face down and up (the \emph{top path}).
We denote the leftmost $n+1$ vertices of degree three on the upper side of the bottom path with $X_1, \ldots, X_n$, and $B_1$.
The leftmost $n+1$ vertices of degree three on the lower side of the top path are denoted $X'_1, \ldots, X'_n$, and $B_2$.
An embedding with total angular resolution $60^{\circ}$ of this frame is again unique up to rotation, scaling, and reflection. 
We assume without loss of generality that it is embedded with $X_1, \ldots, X_n$ being on a horizontal line, as depicted in Figure~\ref{fig:60hard:basic_structure}. 
Then, for every $1\leq i \leq n$, $X'_i$ and $X_i$ lie on a vertical line. 
Further, the line $\ell_1$ spanned by $B_1$ and $C_m$ has slope $60^{\circ}$ and the line $\ell_2$ through $B_2$ and $C_1$ has slope $-60^{\circ}$. 

We next add the variable gadgets in the following way. 
For each variable $x_i$, we identify the vertex $A_{i,1}$ of its gadget with $X_i$.
Further, we connect the gadget to $X'_i$ via a connector gadget by identifying $A_{i,2}$ with $A_{i,3}$ and $A_{i,4}$ with~$X'_i$, respectively.
Note that in any drawing with total angular resolution $60^{\circ}$ of the construction so far, 
each variable gadget together with its connector gadget must be drawn vertically and between $X_i$ and $X'_i$.
Further, the variable gadgets can be scaled by adapting the height of the connector gadget. 
Independent of the scaling factor, the right side of each variable gadget is always to the left of the lines $\ell_1$ and $\ell_2$.  
Directionwise, variable gadgets can be drawn in two ways: either all $X_{i,j}$ are to the right of the $\overline{X}_{i,j}$ or the other way around.

To complete the construction,  
we add a path consisting of three consecutive edges between $X_{i,j}$ ($\overline{X}_{i,j}$) and $C_j$ whenever $x_i$ ($\overline{x}_i$) is a literal of clause $c_j$.
An example of $G$ with the \threesat\ formula $(x_1 \vee \overline{x}_2 \vee x_3) \wedge
(\overline{x}_1 \vee {x}_2 \vee \overline{x}_3) \wedge
(\overline{x}_1 \vee \overline{x}_2 \vee \overline{x}_3) \wedge
(x_1 \vee {x}_2 \vee x_3)$ is depicted in Figure~\ref{fig:60hard:full}.
To obtain a total angular resolution of $60^{\circ}$ at every clause vertex $C_i$, all of these paths must start from $C_i$ towards the right
and one of them must start horizontally. 
We claim that the constructed graph $G$ has a drawing $D$ with $\TAR(D)\geq 60^{\circ}$ if and only if the initial \threesat\ formula is satisfiable.

\begin{figure}[tb]
	\centering
	\subfloat[True connection, two versions\label{fig:60hard:true_connection}]{
			\includegraphics[page=1]{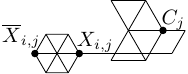}
			\qquad
			\includegraphics[page=2]{figures/hardness_for_geq_60/true_connection}
	}
	\hfill
	\subfloat[False connection\label{fig:60hard:false_connection}]{
		\includegraphics{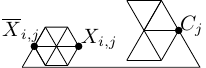}
	}
	\caption{Connections between clause and literal vertices in the \NP-hardness proof.}
	\label{fig:60hard:connections}
\end{figure}

\begin{figure}
  \centering
  \includegraphics{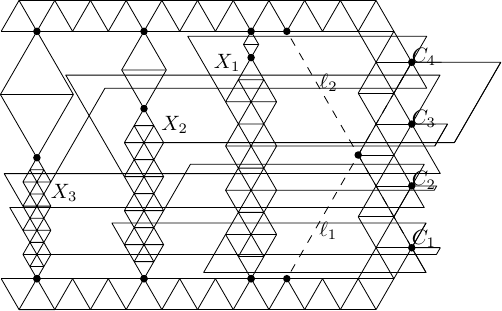}
  \caption{Construction of the graph $G$ for the \threesat\ formula $(x_1 \vee \overline{x}_2 \vee x_3) \wedge
(\overline{x}_1 \vee {x}_2 \vee \overline{x}_3) \wedge
(\overline{x}_1 \vee \overline{x}_2 \vee \overline{x}_3) \wedge
(x_1 \vee {x}_2 \vee x_3)$. For better readability, the variable gadgets are placed with a larger distance between them than in the description of the construction.
The corresponding truth assignment is: $x_1$ is false, $x_2$ is true, $x_3$ is true.}

  \label{fig:60hard:full}
\end{figure}

Assume first that the \threesat\ formula is satisfiable.
Consider a truth assignment of the variables that satisfies the formula.
We draw each variable gadget such that the side corresponding to its true literal is on the right.
Further, we scale all the variable gadgets such that no two vertices of different variable gadgets or of a variable gadget and a clause gadget lie on a common horizontal line (except for the vertices $X_i$ and $X'_i$).
For every clause $c_j$, we choose a literal $v_i \in \{x_i,\overline{x}_i\}$ of $c_i$ which is true. 
We draw the path between the corresponding clause vertex $C_j$ and the matching literal vertex $V_{i,j} \in \{X_{i,j},\overline{X}_{i,j}\}$ in the following way; see Figure~\ref{fig:60hard:true_connection}. We start with a horizontal edge from $C_j$ to the right. Then, we continue with a $\pm 60^{\circ}$ edge towards the left until we reach the height of $V_{i,j}$.
We complete the path with a horizontal edge towards the left to $V_{i,j}$.
For the other literals of $c_j$ we draw a $\pm 60^{\circ}$ edge from $C_j$ to the right, followed by a horizontal edge to the left and a $\pm 60^{\circ}$ edge to the left or right, depending on whether $v_i$ is true or false; see Figure~\ref{fig:60hard:connections}.
This way, all edges of the resulting drawing $D$ are either horizontal or under an angle of $\pm 60^{\circ}$ and no two edges overlap. Hence we have $\TAR(D)=60^{\circ}$ as desired.

For the other direction, assume that this graph $G$ admits a drawing $D$ with ${\TAR(D)=60^{\circ}}$.
In $D$, consider a clause vertex $C_j$ and the path ${P= C_j M_1 M_2 V_{i,j}}$ which starts horizontally towards the right at $C_j$.
Then, the literal vertex $V_{i,j}$ must be on the right side of its variable gadget:
If $V_{i,j}$ is a left vertex of a variable gadget, then $P$ must enter $V_{i,j}$ from the left under an angle of at most $\pm 60^{\circ}$ with respect to the horizontal line.
Hence, $M_2$ lies to the left of the lines $\ell_1$ and~$\ell_2$. 
On the other hand, the second vertex $M_1$ of $P$ lies horizontally to the right of $C_j$.
However, to respect the $60^{\circ}$ restriction at $M_1$, $M_2$ must lie to the right of the lines $\ell_1$ and $\ell_2$, a contradiction.
Now consider the set of literal vertices that are an endpoint of a path starting horizontally at some clause vertex.
As these literal vertices are on the right side of their corresponding variable gadgets, the set does not contain any pair $X_{i,j}, \overline{X}_{i,k}$. 
By setting all the corresponding literals to true, we obtain a non-contradicting 
(possibly partial) truth assignment of the variables which has at least one literal set true for every clause.
Hence, any completion to a truth assignment of all variables satisfies the formula.
\end{proof}

\section{$\TAR$ critical graph}

In Section~\ref{sec:upperbounds} we provided upper bounds on the number of edges a graph with a given total angular resolution can have, where the focus was on $60^{\circ}$ and $90^{\circ}$.
In the previous section, we have seen that deciding if a graph can be drawn with a total angular resolution of a least $60^{\circ}$ is NP-hard.
So naturally it is of central interest to better understant the structure of graphs that do not allow for a certain total angular resolution.
In the following we shed some light on graphs $G$ with $\TAR(G)=60^{\circ}$, such that removing any single edge from $G$ increases its total angular resolution.
In a certain sense these graphs have the minimal structure that forces $\TAR(G)$ to be $60^{\circ}$.
Thus, we call such graphs $\TAR(G)-60^{\circ}$ critical graphs (see below for a proper definition).
A better understanding of these graphs will help to see their structure and why some graphs need $\TAR(G)=60^{\circ}$ while other, very similar, graphs can be drawn with a larger total angular resolution.
We round up this picture in Subsection~\ref{sec:TAR_near_complete} by considering almost complete $\TAR$ critical graphs.

\subsection{$\TARC[60^{\circ}]$ critical graphs}\label{sec:small_graphs_small_angle}

\noindent In this section we give a construction of a family of graphs that have a small number of edges and ${\TAR(G)\leq 60^{\circ}}$.
Since we can construct connected graphs with a small number of edges and $\TAR(G)\leq 60^{\circ}$ by taking a triangle and adding a path to it, we only look at so called \emph{$\TARC$ critical} graphs.
These are connected graphs with $\TAR(G)\leq \alpha$, but for all edges $e \in E$, $\TAR(G\backslash\{e\})>\alpha$ holds.
In other words, $\TAR(G)\leq \alpha$ and every proper subgraph $H$ of a $\TARC$ critical graph $G$ has $\TAR(H)>\alpha$. 
We focus on $\TARC[60^{\circ}]$ critical graphs.

\begin{theorem}\label{thm:60critical}
There exist $\TARC[60^{\circ}]$ critical graphs on $n$ vertices with $\frac{3}{2}n$ edges for infinitely many values of $n$.
\end{theorem}

This means there exist graphs with much fewer than $2n-6$ edges, which have $\TAR(G) \leq 60^{\circ}$.
We prove Theorem~\ref{thm:60critical} by giving a construction of such graphs.
Before we state our construction, we consider two $4$-cycles that share an edge.
\begin{lemma}\label{lemma:4gons}
Two $4$-cycles that share an edge (denoted by $L$) can be embedded with $\TAR(D)>60^{\circ}$ only if $L$ is drawn combinatorially equivalent to $E7$ in Figure~\ref{fig:exceptionsAll}. 
\end{lemma}

\begin{proof}
First note that $L$ has $6$ vertices and $7$ edges.
Hence, $L$ is an exception for Theorem~\ref{theorem:2n-6}.
Therefore, $L$ is drawn like $E7$ in Figure~\ref{fig:exceptionsAll}.
\end{proof}

With the help of Lemma~\ref{lemma:4gons} we show the existence of a graph with $\frac{3}{2}n$ edges and $\TAR(G)\leq 60^{\circ}$. 
To this end, we first define a graph $\MS_k$ consisting of a sequence of 4-cycles glued together along opposite edges,
which is essentially a circularly closed ladder graph on a M\"obius strip.

\begin{definition}
We define the graph $\MS_k$ as follows.
Let $u_i,v_i$, $0\leq i \leq k-1$, be the vertices of $\MS_k$.
Further, we define $u_k=v_0$ and $v_k=u_0$.
The edges of $\MS_k$ are
\begin{itemize}
\item $(u_i,v_i)$ for $0\leq i \leq k$,
\item $(u_i,u_{i+1})$ and $(v_i,v_{i+1})$ for $0\leq i \leq k-1$.
\end{itemize}
\end{definition}
Figure~\ref{fig:60_minimal_family} depicts the graph $\MS_6$, where the dashed and the dashed-dotted edge each highlight one instance of the two different edge types.
Observe that $\MS_k$ has $n=2k$ vertices and $\frac{3}{2}n=3k$ edges.

\begin{figure}
\centering
\includegraphics[page=3]{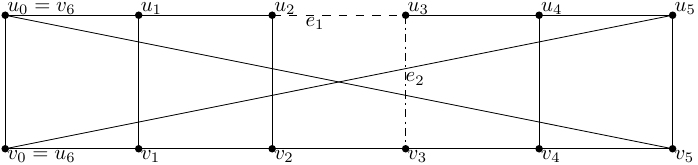}
\caption{Graph $\MS_6$ with $\TAR(\MS_6)=60^{\circ}$ and $\frac{3}{2} n$ edges.}
\label{fig:60_minimal_family}
\end{figure}

\begin{lemma}\label{lemma:MS:critical}
Let $k\geq 6$ be an integer. Then $\MS_k$ is a $\TARC[60^{\circ}]$ critical graph.
\end{lemma}

\begin{proof}
First we show that $\TAR(\MS_k)\leq 60^{\circ}$.
We assume to the contrary that we can embed $\MS_k$ with $\TAR(\MS_k)>60^{\circ}$.
Define $C_i$ as the $4$-cycle $u_i u_{i+1} v_{i+1} v_i$ for $0\leq i \leq k-1$. 
If we can embed $\MS_k$ with $\TAR(\MS_k)>60^{\circ}$, then every $C_i$ is a plane $4$-cycle for every $0\leq i\leq k-1$ and $C_i$ and $C_{i+1}$ are interior disjoint (with $C_{k}=C_0$) due to Lemma~\ref{lemma:4gons}.
First place a point $c_i$ into every $C_i$.
We draw a closed curve~$B$  through all $c_i$ such that between $c_i$ and $c_{i+1}$ in $C_i$ and $C_{i+1}$, $B$ crosses only the edge $u_{i+1}v_{i+1}$ and this edge is crossed only once between $c_i$ and $c_{i+1}$.  

This is possible, because $C_i$ and $C_{i+1}$ are interior disjoint for any $i$.
Since $C_i$ and $C_{i+1}$ are interior disjoint and every $C_i$ is plane, all vertices $u_i$ are on the same side and all the vertices $v_i$ are on the other side of $B$ if we walk along $B$.
But then $u_0$ and $v_k$ are on different sides which gives us the contradiction since $u_0=v_k$.
Therefore $\TAR(\MS_k)\leq 60^{\circ}$.

To show that $\MS_k$ is $\TARC[60^{\circ}]$ critical, we have to prove that for all edges~$e$ ${\TAR(\MS_k\backslash \{e\})>60^{\circ}}$ holds.
In Figure~\ref{fig:60_minimal_family} we see that the graph consists of edges like $e_1$, which are incident to one cycle of length $4$ and edges like~$e_2$, which are incident to two cycles of length $4$.
So we have two cases:
Does $\TAR(\MS_6\backslash \{e_1\})>60^{\circ}$ hold and does $\TAR(\MS_6\backslash \{e_2\})>60^{\circ}$ hold?

\begin{figure}[hbt]
	\centering 
	\subfloat[Drawing of $\MS_6\backslash\{e_1\}$ with $\TAR(\MS_6\backslash \{e_1\})>60^{\circ}$.\label{fig:tar_minimal_e1}]{
		\includegraphics[page=1]{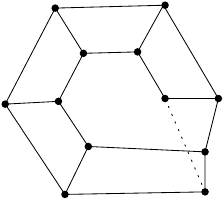}
	}
	\hfill
	\subfloat[Drawing of $\MS_6\backslash\{e_2\}$ with $\TAR(\MS_6\backslash \{e_2\})>60^{\circ}$.\label{fig:tar_minimal_e2}]{
		\includegraphics[page=2]{figures/tar_minimal_embedable}
	}
	\caption{Drawings of $\MS_6$ without one edge.} 
\end{figure}

Figure~\ref{fig:tar_minimal_e1} and Figure~\ref{fig:tar_minimal_e2} depict how ${\MS_6\backslash\{e_1\}}$ and $\MS_6\backslash\{e_2\}$, respectively, can be embedded with ${\TAR(\MS_6\backslash\{e_1\})>60^{\circ}}$.
The dotted edge in both figures is the removed edge.
$\MS_k$ with $k\geq 6$ can be embedded in a similar way (for example, by apropriately subdividing the two opposite edges of a crossing-free 4-cycle that are incident to only that 4-cycle and connecting the subdivision points).
This completes the proof that $\MS_k$ is $\TARC[60^{\circ}]$ critical. 
\end{proof}
Now we have all the results to prove Theorem~\ref{thm:60critical}.
\begin{proof}[Proof of Theorem~\ref{thm:60critical}]
The graph $\MS_k$ with $k\geq 6$ is $\TARC[60^{\circ}]$ critical by Lemma~\ref{lemma:MS:critical}.
Furthermore $\MS_k$ has $\frac{3}{2}n$ edges since it is cubic.
Therefore there exist $\TARC[60^{\circ}]$ critical graphs on $n$ vertices with $\frac{3}{2}n$ edges for infinitely many values $n$.
\end{proof}

\subsection{Almost complete graphs}\label{sec:TAR_near_complete}

\noindent 
Let $K_n$ be the complete graph on $n$ vertices.
Argyriou, Bekos and Symvonis~\cite{Argyriou2013MaximizingTT} showed that $\TAR(K_n)=\frac{180^{\circ}}{n}$.
In this section we show how the deletion of a few edges affects the total angular resolution. 
We start by showing that the removal of a small number of edges does not change the total angular resolution.

\begin{theorem}\label{thm:almost_complete}
Every graph $G$ with $n$ vertices and at least ${n \choose 2}-\frac{n-3}{3}$ edges has $\TAR(G)=\frac{180^{\circ}}{n}$.
\end{theorem}

\begin{proof}
Consider a drawing $D$ of the complete graph $K_n$ with $(n-k)$ vertices on the boundary $B$ of the convex hull and $k$ inner vertices.

A triangle $T$ of $D$ is called special if its vertices are on $B$ and the three inner angles of $T$ are split in total into at least $n$ angles in $D$. 
Note, that the existence of a special triangle implies $\TAR(D)\leq\frac{180^{\circ}}{n}$.
If we delete a set $E$ of at most $\frac{n-k-3}{2}$ edges of $D$, then there are three vertices on $B$ which are not incident to any deleted edge.
So these three vertices span a special triangle of $D\backslash{E}$ where $D\backslash{E}$ is the drawing $D$ without the edges in $E$.

On the other hand, $B$ is an $(n-k)$-cycle and its inner angles sum up to $(n-k-2)180^{\circ}$.
Since we have $K_n$, the inner angles of $B$ are split into $(n-k)(n-2)$ angles.
If the inner angles of $B$ are split into at least $(n-k-2)n$ angles then the total angular resolution is at most $\frac{180^{\circ}}{n}$.
So we can delete up to $\frac{1}{2}((n-k)(n-2)-(n-k-2)n)=k$ edges and still have a drawing with $\TAR(D)\leq\frac{180^{\circ}}{n}$.

Therefore, we want to minimize the maximum of $k$ and $\frac{n-k-3}{2}$ over all possible values of $k$. 
This minimum is obtained for 
	$k=\frac{n-3}{3}$.
	So any graph $G$ with at least $\binom{n}{2}-\frac{n-3}{3}$ edges still has $\TAR(G)=\frac{180^{\circ}}{n}$.
\end{proof}

Starting from the complete graph $K_n$, Theorem~\ref{thm:almost_complete} implies that we 
have to delete more than $\frac{n-3}{3}$ edges to increase the total angular resolution.
On the other hand, we can improve the total angular resolution by deleting $n-2$ edges, which are incident to the same vertex.
This creates a graph $G'$ that is essentially $K_{n-1}$ with an additional vertex connected to the $K_{n-1}$ by a single edge and thus $\TAR(G') = \frac{180^{\circ}}{n-1}> \frac{180^{\circ}}{n}$.
 We now show that the total angular resolution can be increased by removing even fewer edges.

\begin{proposition}\label{thm:complete:all}
For any $n \geq 12$ there exists a graph $G$ with $n$ vertices, at least ${n \choose 2}- \frac{11 n}{12}+1$ edges and $\TAR(G)\geq\frac{180^{\circ}}{n-1}$.
\end{proposition}

\begin{proof}
We take a drawing $D$ of $K_{n-1}$ where the vertices $v_1, v_2, \dots, v_{n-1}$ span a regular $(n-1)$-gon $P$. 
Note that $\TAR(D)=\frac{180^{\circ}}{n-1}$.
Let $c$ be the circumcenter of $P$ and $C$ be the corresponding circumcircle.
Let $p$ be a point on the line spanned by $c$ and $v_1$ such that $|cv_1|=|v_1p|$ as in Figure~\ref{thm:complete:all}.
Observe that the angle $\angle v_ipv_{i+1}<\angle pv_{i+1}v_i$ since $v_iv_{i+1}$ is the shortest edge of the triangle $pv_iv_{i+1}$.

\begin{figure}
\centering
\includegraphics{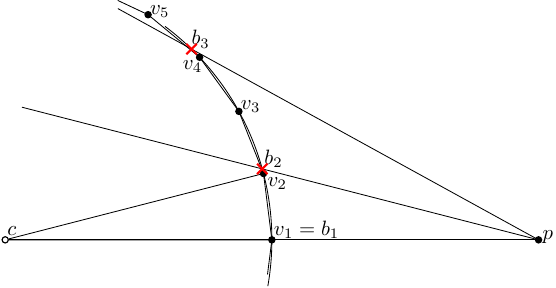}
\caption{Illustration of the proof of Proposition~\ref{thm:complete:all}. Vertices of the graphs are marked with black disks.}
\label{fig:complete_n_even}
\end{figure}

Let the tangents of $C$ through $p$ touch $C$ at the points $t_1$ and $t_2$.
Let $t_1$ be the tangent point which lies on the arc between $v_a$ and $v_{a+1}$ such that $v_a$ is closer to $v_1$ than $v_{a+1}$.
Observe that $\angle v_{i-1}pv_{i} > \angle v_ipv_{i+1}$ holds for $1<i<a$.
Since $|cp|=2\cdot |ct_1|$ and $\angle ct_1p= 90^{\circ}$, the triangle $ct_1p$ is half of an equilateral triangle.
So $\angle pct_1 = 60^{\circ}$ and $\angle cpt_1 = 30^{\circ}$.
Further $\angle cv_1v_2 < 90^{\circ}$ holds and therefore $\angle v_1pv_2 < \angle v_2cv_1 = \frac{360^{\circ}}{n-1}$.
So $\angle v_ipv_{i+1}<\frac{360^{\circ}}{n-1}$ holds for $1\leq i<a$.

We place points $b_j$ on $C$ on the shorter arc between $v_1$ and $v_a$ as depicted in Figure~\ref{fig:complete_n_even} such that
$b_1$ is $v_1$ and $\angle b_j p b_{j+1} = \frac{360^{\circ}}{n-1}$ for $1\leq j \leq \lfloor \frac{n-1}{12} \rfloor$.
Note that the bound $\lfloor \frac{n-1}{12} \rfloor$ is implied by $\angle cpt_1 = 30^{\circ}$.
In Figure~\ref{fig:complete_n_even} the points $b_2$ and $b_3$ are marked with a red cross.
Since $\angle v_ipv_{i+1}<\frac{360^{\circ}}{n-1}$ for $1\leq i<a$, there is an index $k$ with $1<k<a$ such that $v_k$ is on the arc between $b_j$ and $b_{j+1}$ including $b_{j+1}$ for all $1\leq j \leq \lfloor \frac{n-1}{12} \rfloor$.

Let $S_1$ be a set of vertices of $P$ such that for any $1\leq j \leq \lfloor \frac{\lfloor \frac{n-1}{12} \rfloor}{2} \rfloor$ there exists exactly one point $v_k$ of $S_1$ such that $v_k$ is on the arc from $b_{2j}$ to $b_{2j+1}$.
So $S_1$ contains $\lfloor \frac{\lfloor \frac{n-1}{12} \rfloor}{2} \rfloor$ points.
Note that $v_1$ is not in $S_1$.

Let $S_2$ be the set of vertices of $P$ such that $v_{n+1-i} \in S_2$ if and only if $v_i \in S_1$.
Let $S= S_1 \cup S_2 \cup \{v_1\}$.
If $v_i,v_k \in S_1$ or $v_i,v_k \in S_2$, then $\angle v_ipv_k \geq \frac{360^{\circ}}{n-1}$ holds by construction.
If $v_i \in S_1$, then $\angle v_1pv_i \geq \frac{360^{\circ}}{n-1}$ holds since $\angle v_1pv_i > \angle v_1pb_2$. If $v_i \in S_2$, then $\angle v_1pv_i \geq \frac{360^{\circ}}{n-1}$ follows in a similar way.
Therefore, for any two  points $v_i,v_k \in S$, $\angle v_ipv_k \geq \frac{360^{\circ}}{n-1}$ holds.

So the graph $G$ with vertices $v_1, \dots, v_{n-1}$ and $p$, and edges $v_iv_j$ for any $1\leq i,j \leq n-1$ and $pv_i$ for $v_i \in S$ has $\TAR(G) \geq \frac{180^{\circ}}{n-1}$ and ${n-1 \choose 2} + 2 \lfloor \frac{\lfloor \frac{n-1}{12} \rfloor}{2} \rfloor +1$ edges.
Hence, $G$ has at least ${n \choose 2}- \frac{11 n}{12}$ edges.
\end{proof}

Proposition~\ref{thm:complete:all} holds for all $n\geq 12$ but is most likely not tight for infinitely many values of $n$.
If the number of vertices is odd, then the following proposition gives us a better bound for the number of edges.

\begin{proposition}\label{thm:complete:odd}
For any odd $n\geq 5$ there exists a graph $G$ with $n$ vertices, ${n \choose 2}-\frac{n-1}{2}$ edges and $\TAR(G)\geq\frac{180^{\circ}}{n-1}$.
\end{proposition}

\begin{proof}
This is achieved by the following construction. 
We take a drawing $D$ of $K_{n-1}$ where the vertices $v_1, v_2, \dots, v_{n-1}$ form a regular $(n-1)$-gon. 
Next, we replace the common crossing of all main diagonals $(v_i, v_{i+(n-1)/2})$, $1\leq i \leq (n-1)/2$, by a vertex $v_n$.
We also replace every main diagonal $(v_i, v_{i+(n-1)/2})$ by the edges $(v_i, v_n)$ and $(v_n, v_{i+(n-1)/2})$ for every $1\leq i \leq \frac{n-1}{2}$. 
We denote the resulting drawing with $D'$. Figure~\ref{fig:complete_n_odd} depicts $D'$ for $n=9$ vertices.

\begin{figure}
\centering
\includegraphics{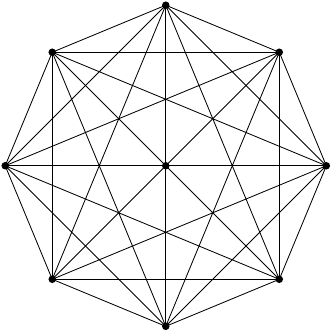}
\caption{A drawing with $32$ edges and $\TAR(D)=22.5^{\circ}$.}
\label{fig:complete_n_odd}
\end{figure}

Since we only replaced edges and do not have edges which are on top of each other, we have $\TAR(D')=\TAR(D)=\frac{180^{\circ}}{n-1}$.
Further, $D'$ has $\frac{(n-1)^2}{2}$ edges. So $D'$ has $\frac{n-1}{2}$ edges fewer than $K_n$.
\end{proof}

We have shown that in a complete graph $K_n$ we can delete $\frac{11 n}{12}-1$ edges for arbitrary $n$ and $\frac{n-1}{2}$ edges for odd $n$ to increase the total angular resolution.
For odd $n$, deleting any $\frac{n-3}{3}$ edges does not affect the total angular resolution but deleting $\frac{n-1}{2}$ edges can improve it.
We conjecture that Proposition~\ref{thm:complete:odd} is tight.


\section{Some special graphs}

\noindent In this section we present some special graphs and their total angular resolution. 
\subsection{The Petersen graph}

\noindent First we study the Petersen graph.
\begin{theorem}
Let $G_P$ be the Petersen graph.
Then $\TAR(G_P)=60^{\circ}$.
\end{theorem}

\begin{proof}
$G_P$ has 10 vertices and 15 edges.
By Theorem~\ref{theorem:2n-6}, since $m>2n-6$ and $G_P$ is not in the exceptions, we have $\TAR(G_P)\leq 60^{\circ}$.
On the other hand, two drawings of $G_P$ with total angular resolution of $60^{\circ}$ are shown in Figure~\ref{subfigure:petersen_tar}.
This means that we have $\TAR(G_P)=60^{\circ}$.
\end{proof}

\begin{figure}[hbt]
	\centering 
	\subfloat[Classic drawing of the Petersen graph\label{subfigure:petersen}]{
		\includegraphics[page=1]{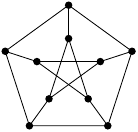}
	}
	\hspace{10mm}
	\subfloat[Drawings of the Petersen graph with $\TAR(D)=60^{\circ}$\label{subfigure:petersen_tar}]{
		\includegraphics[page=2]{figures/petersen_graph}
	}
	\caption{Drawings of the Petersen graph.} 
\end{figure}

The Petersen graph was generalized in the following way~\cite{Coxeter1950Self}.
\begin{definition}
	Given two integers $n \geq 3$ and $k$ with $1\leq k<\frac{n}{2}$, the \emph{generalized Petersen graph} $\GPG(n,k)$ is defined as follows.
Denote the vertices of the graph $u_0,u_1,\dots,u_{n-1},v_0,v_1,\dots,v_{n-1}$.
Then the edges of\, $\GPG(n,k)$ are $u_i u_{i+1}$, $u_i v_i$ and $v_i v_{i+k}$ for $i=0,\ldots,n-1$, where $v_{i+k}:=v_{i+k-n}$ if $i+k>n-1$ and  $u_{n}:=u_{0}$.
\end{definition}

\begin{figure}
\centering
\includegraphics{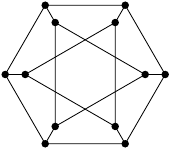}
\label{fig:gpg6_2}
\caption{Classic drawing of $\GPG(6,2)$.}
\end{figure}

Note that the Petersen graph is the graph $\GPG(5,2)$.
The question is to determine for which values of $n$ and $k$ the generalized Petersen graph admits a drawing with larger total angular resolution than the Petersen graph.
Note that $\TAR(\GPG(n,n/3))\leq 60^{\circ}$ since $\GPG(n,n/3)$ contains triangles. Especially, $\TAR(\GPG(6,2)) = 60^{\circ}$ is obtained by the classic drawing of $\GPG(6,2)$.
In the following lemma we show that $\TAR(\GPG(n,2))>60^{\circ}$ for any $n\geq 7$.

\begin{proposition}
	For any $n\geq 7$, the generalized Petersen graph $\GPG(n,2)$ has $\TAR(\GPG(n,2))>60^{\circ}$.
\end{proposition}

\begin{figure}
	\subfloat[Drawing $D(7)$ of $\GPG(7,2)$ with ${\TAR(D(7))>60^{\circ}}$\label{subfigure:petersen7}]{
		\includegraphics[page=7]{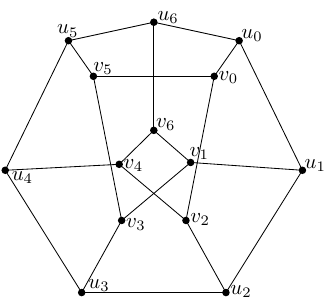}
	}
	\hspace{10pt}
	\subfloat[Construction of the Drawing $D(9)$ of $\GPG(9,2)$\label{subfigure:petersen7_2}]{
		\includegraphics[page=6]{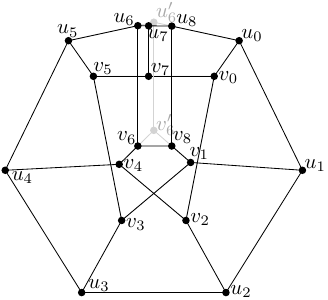}
	}
	\caption{Drawings of $\GPG(n,2)$ with  odd $n$ and $\TAR(D)>60^{\circ}$}
	\label{fig:genpeter}
\end{figure}

\begin{proof}
We give a construction for a drawing $D(n)$ of $\GPG(n,2)$ for every $n\geq 7$  with $\TAR(D(n))>60^{\circ}$.
The drawing depends on whether $n$ is even or odd.
\begin{description}
\item[Case 1: $n\geq 7$ is odd.]
Figure~\ref{subfigure:petersen7} depicts a drawing $D(7)$ of $\GPG(7,2)$ with $\TAR(D(7))>60^{\circ}$.
For any odd $n>7$, we modify $D(7)$ to a drawing~$D(n)$, as illustrated in Figure~\ref{subfigure:petersen7_2} for $D(9)$, such that $D(n)$ is a drawing of $\GPG(n,2)$ with $\TAR(D(n))=\TAR(D(7))>60^{\circ}$.
Note, that the edges $u_6v_6$ and $v_5v_0$ cross with a $90^{\circ}$ angle in $D(7)$.

First we take $D(7)$ and rename the point $u_6$ to $u'_6$ and $v_6$ to $v'_6$.
We place the point $v_6$ in the interior of the line segment $v_4v'_6$.
Then we place $u_6$ on $u_5u'_6$ such that $u_6v_6$ is parallel to $u'_6v'_6$
and  $v_{n-1}$ on $v_1v'_6$ such that $v_6v_{n-1}$ is parallel to $v_5v_0$.
Further we place $u_{n-1}$ on $u_0u'_6$ such that $u_6u_{n-1}$ is parallel to $v_5v_0$.

For $7\leq i \leq n-2$ we place points $u_i$ in the interior of the line segment $u_6u_{n-1}$ from left to right.
For odd $i$ we place $v_i$ onto $v_5v_0$ such that $u_iv_i$ is parallel to $u'_6v'_6$.
For even $i$ we place $v_i$ onto $v_6v_{n-1}$ such that $u_iv_i$ is parallel to $u'_6v'_6$.
Then the vertices $u_i$ and $v_i$, $0\leq i \leq n-1$, span the drawing $D(n)$ of $\GPG(n,2)$.

All angles of $D(n)$ that already appear in $D(7)$ are greater than $60^{\circ}$.
Further, all angles that are in the subdrawing induced by the vertices $u_i$ and $v_i$ with $6\leq i \leq n-1$ have $90^{\circ}$ by construction.
Also the angles incident to one of the vertices $u_6, v_6, u_{n-1}, v_{n-1}$ are greater than $60^{\circ}$.
Therefore, $\TAR(D(n))=\TAR(D(7))>60^{\circ}$.

\item[Case 2: $n\geq 8$ is even.]
If $n$ is even, then $\GPG(n,2)$ contains the two cycles $v_0 v_2 \dots v_{n-2}$ and $v_1 v_3 \dots v_{n-1}$ which are both of length $\frac{n}{2}$.
We draw $u_0,u_1,\dots,u_{n-1}$ as a regular polygon $P_u$.
Let $c$ be the center of $P_u$. 
We place the points $v_i$, with even $i$ 
inside $P_u$ such that for each $i\in\{0, 2, 4,\dots, n-2\}$, $v_i$ is on the line segment $u_ic$, and the length of all such line segments $u_iv_i$ is identical. 
We then place the points $v_i$ with odd $i$ 
inside the polygon spanned by $v_0,v_2,v_4,\dots,v_{n-2}$ such that for each $i\in\{1, 3, 5,\dots, n\!-\!1\}$, $v_i$ is on the line segment $u_ic$,
and the length of all such line segments 
$u_iv_i$ is again identical. 
Denote the resulting drawing by $D(n)$; see Figure~\ref{fig:gp8-2} for a depiction of $D(8)$. 
\begin{figure}
	\centering
	\includegraphics[page=2]{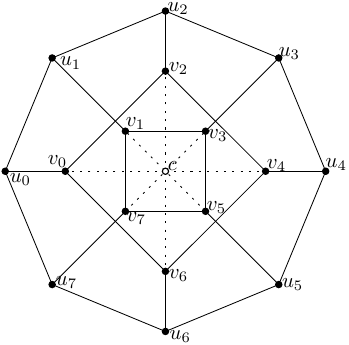}
	\caption{Drawing of $\GPG(8,2)$ with $\TAR(D)=67,5^{\circ}$}
	\label{fig:gp8-2}
\end{figure}
In $D(n)$, every crossing is between two line segments $u_iv_i$ and $v_{i-1}v_{i+1}$ for odd $i$ with $1\leq i \leq n-1$ where $v_n=v_0$. 
Since $u_{i-1},v_{i-1}$, and $c$ are collinear, and since $|u_{i-1}c|=|u_{i+1}c|$ and $|u_{i-1}v_{i-1}|=|u_{i+1}v_{i+1}|$, also $|v_{i-1}c|=|v_{i+1}c|$ holds.
Hence $v_{i-1}v_{i+1}c$ is an isosceles triangle.
The line spanned by $u_i$ and $v_i$ is the angle bisector of the angle $\angle v_{i-1}cv_{i+1}$.
Therefore, the angles at the crossing of $u_iv_i$ and $v_{i-1}v_{i+1}$ are $90^{\circ}$.
This means that an angle in $D(n)$ with at most $60^{\circ}$ cannot be at a crossing of the drawing.
It remains to consider the angles at each vertex. The smallest such angles appear at the vertices $u_i$ and have size $\frac{90^{\circ} (n-2)}{n}$.
It follows, that $\TAR(D(n))\geq \frac{90^{\circ} (n-2)}{n} > 60^{\circ}$ since $n\geq 8$.
\end{description}
\end{proof}

\subsection{Graphs where crossings help to improve}

\noindent In this section we present some planar graphs $G$ where any drawing $D_G$ with $\TAR(D_G)=\TAR(G)$ is not a plane drawing. Van Kreveld~\cite{van2010quality} already showed that the angular resolution of RAC drawings is in some cases better than the angular resolution of plane drawings. 
Our main interest for this section is to find a graph, where $\TAR(D_G)\leq 60^{\circ}$ if $D_G$ is plane but $\TAR(G)>60^{\circ}$.

\begin{figure}[hbt]
	\centering 
	\subfloat[Plane $K_4$ with $\TAR(D)=30^{\circ}$\label{subfigure:K4:plane}]{
		\includegraphics[page=1]{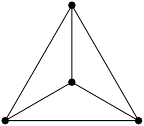}
	}
	\hspace{50pt}
	\subfloat[$K_4$ with crossing and $\TAR(D)=45^{\circ}$\label{subfigure:K4:cross}]{
		\includegraphics[page=2]{figures/K4}
	}
	\caption{Two drawings of the $K_4$.} 
\end{figure}

If we look at the complete graph $K_4$, we see in Figure~\ref{subfigure:K4:plane}, that every plane drawing has $\TAR(D)\leq 30^{\circ}$.
When we place the points as vertices of a square like in Figure~\ref{subfigure:K4:cross}, then we have $\TAR(D)=45^{\circ}$.
Therefore, we get in this case a larger $\TAR(D)$ if $D$ has crossings.
Another example can be seen in Figure~\ref{figure:8gon}.
However, neither of these examples satisfies $\TAR(G)>60^{\circ}$ and $\TAR(D)\leq 60^{\circ}$ if $D$ is plane.

\begin{figure}[hbt]
	\centering 
	\subfloat[Plane drawing with $\TAR(D)=62.5^{\circ}$.\label{subfigure:8gon:plane}]{
		\includegraphics[page=2]{figures/8gon_diag}
	}
	\hspace{50pt}
	\subfloat[$8$-cycle with crossing diagonals and $\TAR(D)=90^{\circ}$.\label{subfigure:8gon:cross}]{
		\includegraphics[page=1]{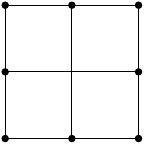}
	}
	\caption{$8$-cycle with diagonals.}\label{figure:8gon}
\end{figure}

\begin{figure}[hbt]
	\centering 
	\subfloat[Plane drawing.\label{subfigure:cross:good_plane}]{
		\includegraphics[page=3]{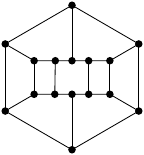}
	}
	\hspace{50pt}
	\subfloat[Drawing with $\TAR(D)=67.5^{\circ}$.\label{subfigure:cross:good_cross}]{
		\includegraphics[page=2]{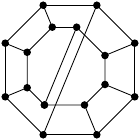}
	}
	\caption{Drawings of two nested $8$-cycles with two diagonals} 
\end{figure}

The graph in Figure~\ref{subfigure:cross:good_cross} has $\TAR(G)\geq 67.5^{\circ}$.
Every face of a plane embedding $D$ of $G$, which can be seen in Figure~\ref{subfigure:cross:good_plane}, has at most $6$ edges.
Since the graph is $3$-regular and the outer face has at most $6$ edges, Lemma~\ref{obs:inner_degree} gives us $\TAR(D)\leq 60^{\circ}$ for every plane drawing of the graph of Figure~\ref{subfigure:cross:good_plane}.


\section{Conclusion and open problems}

\noindent In this work we have shown that, up to a finite number of well specified exceptions of constant size, any graph $G$ with $\TAR(G)>60^{\circ}$ has at most $2n-6$ edges. 
For larger angles we were able to obtain similar bounds:
For graphs with $\TAR(G)\geq 90^{\circ}$ we have $m \leq 2n - 2\sqrt{n}$, for $\TAR(G)> 90^{\circ}$ we have $m \leq \frac{3}{2}n - \frac{5}{2}$, and for $\TAR(G)>120^{\circ}$ we have $m \leq n$ for $n\geq 7$. These bounds are tight.
We conjecture that almost all graphs with $\TAR(G)>\frac{k-2}{k}90^{\circ}$ have at most $2n-2-\lfloor \frac{k}{2} \rfloor$ edges.

From a computational point of view, we have proven that deciding whether a given graph admits a drawing with total angular resolution at least $60^{\circ}$ is in general \NP-hard.
The same was known before for at least $90^\circ$~\cite{Formann1993DrawingGI}.
On the other hand, for large angles, the recognition problem eventually becomes easy (for example, $G$ can be drawn with $\TAR(G) > 120^\circ$ if and only if it is the union of cycles of at least 7 vertices and arbitrary paths). This yields the following open problem: At which angle(s) does the decision problem change from NP-hard to polynomial-time solvable?

We introduced $\TARC$ critical graphs 
and showed the existence of $\TARC[60^{\circ}]$ critical graphs with $\frac{3}{2}n$ edges. It remains open whether there are $\TARC[60^{\circ}]$ critical graphs with fewer than $\frac{3}{2}n$ edges. More generally, how many edges does the smallest $\TARC$ critical graph have for a fixed $\alpha$? It is also open, for which values of $\alpha$ there exist $\TARC$ critical graphs with more than $n$ vertices, where $n$ is arbitrarily large.
For the complete graph $K_n$ we proved that we can delete any $\frac{n-1}{3}$ edges of $K_n$ and still get $\TAR(G)=\frac{180^{\circ}}{n}$. It is open whether this bound is tight. On the other hand we presented two families of drawings, which have $\TAR(G)>\frac{180^{\circ}}{n}$ and many edges.
As a related question, what is the smallest number of edges a graph with $n$ vertices and $\TAR(G)=\frac{180^{\circ}}{n}$ can have?

We showed that the Petersen graph has ${\TAR(G_P)=60^{\circ}}$.
For the generalized Petersen graphs we showed that $\TAR(\GPG(n,2))>60^{\circ}$ for $n>6$.
 
\section*{Acknowledgments}

\noindent This work started during the Japan-Austria Joint Seminar \emph{Computational Geometry Seminar with Applications to Sensor Networks} supported by the Japan Society for the Promotion of Science (JSPS) and the Austrian Science Fund (FWF) under grant AJS 399. We would like to thank all participants for creating a stimulating research environment.
Oswin Aichholzer, Irene Parada, Daniel Perz, and Birgit Vogtenhuber were partially supported by the FWF grants W1230 (Doctoral Program Discrete Mathematics) and I 3340-N35 (Collaborative DACH project \emph{Arrangements and Drawings}).
Yoshio Okamoto was partially supported by JSPS KAKENHI Grant Numbers JP20H05795 and JP20K11670.
An extended abstract of parts of this work has been presented at the 27th International Symposium on Graph Drawing and Network Visualization (2019)~\cite{aichholzer2019graphs}.

\bibliographystyle{abbrv}
\bibliography{TARsources}

\end{document}